\documentclass[11pt]{article}
\usepackage{amsmath}
\usepackage{amsthm}
\usepackage{tikz}
\usepackage{amssymb,float}
\usepackage[utf8]{inputenc}
\usepackage[T1]{fontenc}
\usepackage{caption}
\usepackage{algorithm,mathtools}
\usepackage{algpseudocode}
\usepackage{subcaption}
\usepackage{mathtools, mathabx} 
\usepackage{thm-restate}
\usepackage{fullpage}
\usepackage[margin=1in]{geometry}\usepackage{complexity}
\usepackage{xcolor,hyperref}
\usepackage{tcolorbox}
\usepackage{graphicx}
\newtheorem{theorem}{Theorem}
\newtheorem{observation}[theorem]{Observation}
\newtheorem{lemma}[theorem]{Lemma}
\newtheorem{claim}[theorem]{Claim}
\newtheorem{subclaim}[theorem]{{Sub-Claim}}

\theoremstyle{definition}

\newcommand{\ncna}{{\mathbb{F}_{\bar{A},\bar{C}}[X]}}
\newcommand{\cna}{{\mathbb{F}_{\bar{A},C}[X]}}

\newtheorem{defn}{Definition}

\newcommand{\F}{\mathbb{F}}

\newcommand{\pit}{\mbox{\rm PIT}}

\renewcommand{\int}{\rm{\mbox{int}}}

\renewcommand{\geq}{\geqslant}

\renewcommand{\leq}{\leqslant}

\title{Efficient Polynomial Identity Testing Over Nonassociative Algebras}
\date{}

\author{
C. Ramya\thanks{The Institute of Mathematical Sciences (a CI of Homi Bhabha National Institute), Chennai, India, \texttt{email: ramyac@imsc.res.in.}}
\and
Partha Mukhopadhyay\thanks{Chennai Mathematical Institute, Chennai, India \texttt{email: partham@cmi.ac.in.}}
\and
Pratik Shastri\thanks{The Institute of Mathematical Sciences (a CI of Homi Bhabha National Institute), Chennai, India, \texttt{email: pratiks@imsc.res.in.}}
}

\begin{document}

\maketitle

\begin{abstract}
We design the first efficient polynomial identity testing algorithms over the \emph{nonassociative} polynomial algebra. In particular, multiplication among the formal variables is commutative but it is not associative. This complements the strong lower bound results obtained over this algebra by Hrubeš, Yehudayoff, and Wigderson \cite{HWY10} and Fijalkow, Lagarde, Ohlmann, and Serre \cite{FLOS21} from the identity testing perspective. Our main results are the following:
\begin{itemize}
\item We construct nonassociative algebras (both commutative and noncommutative) which have no low degree identities. As a result, we obtain the first Amitsur-Levitzki type theorems \cite{AL50} over nonassociative polynomial algebras. As a direct consequence, we obtain randomized polynomial-time black-box $\pit$ algorithms for nonassociative polynomials which allow evaluation over such algebras. 

\item On the derandomization side, we give a deterministic polynomial-time identity testing algorithm for nonassociative polynomials given by arithmetic circuits in the white-box setting. Previously, such an algorithm was known with the additional restriction of noncommutativity \cite{ADMR17}. 

\item In the black-box setting, we construct a hitting set of quasipolynomial-size for nonassociative polynomials computed by arithmetic circuits of small depth. Understanding the black-box complexity of identity testing, even in the randomized setting, was open prior to our work.  
\end{itemize}
\end{abstract}


\section{Introduction}\label{sec:intro}

The goal of algebraic complexity is to study the complexity of computing multivariate polynomials using basic arithmetic operations, such as addition and multiplication. Arithmetic circuits and formulas are among the well-studied  models in this area.
In particular, arithmetic circuits are directed acyclic graphs whose leaves are labeled by variables or constants, and whose internal nodes (called gates) are labeled with either $+$ or $\times$. Formulas are circuits whose underlying graph is a tree. Each gate in a circuit computes a polynomial in the natural way and the output of a circuit is said to be the polynomial computed at a distinguished output gate. 

There are two central problems studied in algebraic complexity: one is the question of proving size lower bounds for circuits computing explicit polynomials, and the other is the question of derandomizing polynomial identity testing ($\pit$ for short). The $\pit$ problem comes in two variants. First is the black-box setting, where we are given evaluation access to a circuit (given as a black box), and we must decide whether the polynomial it computes is identically zero. The second, seemingly easier is the white-box setting, where we are explicitly given the circuit as a graph to determine whether it computes the zero polynomial.

Baur and Strassen \cite{BS83} proved that any circuit computing $\sum_{i=0}^nx_i^n$ requires size $\Omega(n\log n)$. This is the strongest known lower bound for arithmetic circuits. On the other hand, the $\pit$ problem admits a randomized polynomial-time black-box algorithm over the usual polynomial ring $\mathbb{F}[x_1, \ldots, x_n]$, thanks to the Polynomial Identity Lemma \cite{Zip79, Sch80, DL78} (see Lemma \ref{lem:SZ} for the exact statement). Here, $\F$ can be any field of sufficient size, and the variables $x_1, \ldots, x_n$ commute under multiplication. For the black-box case, the derandomization problem is essentially equivalent to the efficient construction of small {\em hitting sets}. Despite intense efforts over many years, proving strong lower bounds for circuits and derandomizing $\pit$ for circuits have both remained elusive goals. For more on these problems, the reader is referred to the excellent surveys by Shpilka and Yehudayoff \cite{SY10} and Saxena \cite{Sax14}.
In general, the problems of proving lower bounds and derandomizing $\pit$ are closely related and, in fact, nearly equivalent due to an influential result of Impagliazzo and Kabanets \cite{KI04}. For example, very recently, a subexponential-size hitting set for $\pit$ of low-depth arithmetic circuits was obtained via a breakthrough lower bound result by Limaye, Srinivasan, and Tavenas \cite{LST21}.

The limitation in our understanding of $\pit$ and lower bounds stems from the difficulty of analyzing how a monomial gets canceled at an intermediate step in an arithmetic circuit computation. In the setting of the usual polynomial ring $\F[x_1, \ldots, x_n]$, multiplication is commutative and associative. In other words, for all $x_i, x_j, x_k$ we have $x_ix_j=x_jx_i$ and $(x_i(x_jx_k))=((x_ix_j)x_k)$.
A central line of investigation studies arithmetic circuits by restricting the relations satisfied by the multiplication rule. The hope is that, in the absence of such relations, we can better understand \emph{cancellations} of monomials and gain insight into general circuits. For example, we can drop commutativity while preserving associativity.

If we drop commutativity, we obtain the noncommutative polynomial ring $\mathbb{F}\langle x_1, \ldots,x_n\rangle$ (which remains associative), and noncommutative circuits compute polynomials in the ring $\mathbb{F}\langle x_1, \ldots,x_n\rangle$. In his pioneering work \cite{Nis91}, Nisan proved that the noncommutative Permanent and Determinant polynomials require exponential-size noncommutative algebraic branching programs (ABPs for short). ABPs are a subclass of circuits. Up to polynomial blowup, they simulate formulas and are simulated by circuits. In the noncommutative setting, exponential separations are known between ABPs and circuits \cite{Nis91}. On the $\pit$ side, Raz and Shpilka \cite{RS04} developed a white-box deterministic polynomial-time algorithm for polynomials computed by noncommutative ABPs. In fact, a quasipolynomial time black-box $\pit$ algorithm has also been designed for the same problem by Forbes and Shpilka \cite{FS13}. 

However, if we look at \emph{general} noncommutative circuits (instead of ABPs), again we see that progress on the questions of lower bounds and derandomizing $\pit$ has remained elusive. In fact, for general circuits, the best lower bound and $\pit$ results over $\mathbb{F}\langle x_1,\ldots, x_n \rangle$ match those over $\mathbb{F}[x_1 \ldots, x_n]$.

Henceforth, we use $X$ to denote the set $\{x_1, \ldots, x_n\}$. $\F_{A,C}[X]$ and $\F_{A,\bar{C}}[X]$ stand for the rings $\F[X]$ and $\F\langle{X}\rangle$, respectively. Instead of dropping commutativity, we may choose to drop associativity. This leads us to the algebra $\F_{\bar{A},C}[X]$, the polynomial algebra where multiplication is commutative but nonassociative\footnote{Formally, an {\em algebra} $\mathbb{A}$ over a field $\mathbb{F}$ is a vector space equipped with a bilinear product.} \footnote{In $\cna$, $(x_i(x_j x_k)) \neq ((x_i x_j)x_k)$ even in the case when $x_i = x_j =x_k$.}. If we drop \emph{both} commutativity and associativity, we obtain the polynomial algebra $\ncna$. On the lower bounds side, very impressive progress has been made in the algebra $\F_{\bar{A},C}[X]$.
Hrubeš, Yehudayoff, and Wigderson \cite{HWY10} proved the first exponential-size lower bounds for circuits computing explicit polynomials in $\cna$. Subsequently, Fijalkow, Lagarde, Ohlmann, and Serre \cite{FLOS21} strengthened this result by providing an \emph{exact} characterization of the size of a minimal circuit computing a polynomial $f\in \cna$ in terms of the rank of a certain matrix of coefficients.

On the other hand, in the context of $\pit$, it is imperative to note that no efficient algorithm is known over the algebra $\cna$, not even a randomized white-box algorithm. This is surprising given the intimate connections between lower bounds and PIT in various settings. Over the algebra $\ncna$, Arvind et al.\ designed a deterministic \emph{white-box} polynomial-time algorithm for $\pit$ of arithmetic circuits \cite{ADMR17}. Furthermore, as mentioned by the authors in \cite{ADMR17}, a black-box algorithm is not known even over $\ncna$. Notably, over the algebra $\F\langle{X}\rangle$, such a randomized $\pit$ algorithm is known due to the Amitsur-Levitzki Theorem \cite{AL50, BW05} (see Theorem \ref{thm:amitsur-Levitzki} for a formal statement). 

We note that nonassociative computations are fundamental even beyond algebraic complexity. For example, the composition of operations in computer programs is typically nonassociative.  
As a concrete algorithmic example, in a seminal work, Valiant designed a sub-cubic algorithm for recognizing Context Free Languages (CFL) \cite{Val75}. In particular, Valiant developed an algorithm to compute the transitive closure of upper triangular matrices whose entries are elements of a nonassociative monoid \cite{Val75}. In algebraic complexity, lower bounds for nonassociative circuits have been used prove lower bounds for related \emph{associative} models of computation \cite{FLOS21}.

The main technical contributions of this paper are as follows. Over the polynomial algebra $\cna$, we design a white-box deterministic polynomial-time algorithm for the $\pit$ of arithmetic circuits. In the black-box setting, we develop the first randomized polynomial-time $\pit$ algorithm for nonassociative arithmetic circuits (over both $\cna$ and $\ncna$). This is achieved by proving analogues of the Amitsur-Levitzki theorem over the nonassociative polynomial algebras $\ncna$ and $\cna$. Moreover, for the classes of circuits over the algebras $\cna$ and $\ncna$ with polylogarithmic depth, we construct quasipolynomial-size hitting sets over nonassociative algebras of small dimension. To the best of our knowledge, this is the first black-box derandomization result for a well-studied circuit class over a nonassociative polynomial algebra. 
We elaborate on our results and techniques in the next section.

\subsection{Our Results}\label{sec:ourresult}

In this paper, we complement the strong lower bounds results obtained over the algebra $\cna$ (\cite{HWY10, FLOS21}) by designing efficient PIT algorithms. 
Our results work over all fields of sufficiently large size.

\subsubsection{Black-box randomized nonassociative $\pit$ (Section \ref{Randomized})}

Our first result is a black-box randomized polynomial-time identity testing algorithm for polynomials in $\cna$. Of course, to capture nonassociativity, the natural idea is to evaluate the given polynomial over nonassociative algebras. 

In mathematics, nonassociative algebras are very well-studied. For example, one can see the classic work of Albert \cite{JZBSO93}.
In particular, there are various specific algebras of interest which are nonassociative but commutative. The most important such algebras are, perhaps, the Jordan Algebras. Unfortunately, every Jordan Algebra $\mathbb{J}$ satisfies the \emph{Jordan Identity}: $\forall a, b\in \mathbb{J}$, we have $(ab)(aa)-(a(b(aa)))=0$. This means that performing PIT using Jordan Algebras is not possible. To see why, suppose we are given a circuit computing a non-zero polynomial in the ideal of $\cna$ generated by $\{(x_ix_j)(x_ix_i)-(x_i(x_j(x_ix_i)))\mid x_i,x_j\in X \}$. On any input from $\mathbb{J}$, the circuit will evaluate to $0$, whereas the circuit computes a non-zero polynomial. 
Even over $\ncna$, we have a similar difficulty. For example, Matrix Lie algebras are well studied algebras that are nonassociative and noncommutative, but they satisfy the Jacobi Identity \cite{Bour89}. 

Thus, in order to obtain a fast black-box algorithm, we need to construct a suitable nonassociative algebra that does not satisfy low (in terms of its dimension) degree identities. We should remark that over the noncommutative polynomial ring $\F\langle X\rangle$ (equivalently, $\F_{A, \bar{C}}[X]$), the Amitsur-Levitzki theorem allows us to do black-box randomized $\pit$ for polynomials of degree $\leq d$, using random matrices of dimension $(\lceil d/2\rceil+1)\times (\lceil d/2+1\rceil)$ \cite{AL50,BW05}. In our setting of polynomials in $\cna$, we construct a unital nonassociative, commutative algebra $\mathbb{C}_d$ of dimension $d(d+1)^2+1$ which does not satisfies \emph{any} identity of degree $\leq d$.

\begin{restatable}{lemma}{nonassocAL}
\label{lowerBound2}
    
    Let $f\in \cna$ be a non-zero polynomial of total degree $\leq d$. Then $f$ is not a polynomial identity (PI) for $\mathbb{C}_d$.

\end{restatable}

To the best of our knowledge, this is the first Amitsur-Levitzki type theorem in the nonassociative setting. As an immediate consequence, we obtain a nonassociative analogue of the Polynomial Identity Lemma. 

\begin{restatable}{theorem}{commnonassocSZ}
\label{Comm-NonAssociative-SZ}
    Let $\mathbb{F}$ be a field with $|\mathbb{F}|>d$, and $S \subset \mathbb{F}$. Let $f \in \cna$ be a non-zero polynomial of degree $\leq d$ given as a black-box with query access to evaluations of $f$ on elements of $\mathbb{C}_d$. 
    Let $S\subseteq \mathbb{F}$ with $|S|>d$. Sample $b_1, \ldots, b_n\in \mathbb{C}_{d}$ as follows: Pick each of the $d(d+1)^2+1$ entries of each of the $b_i$'s uniformly and independently from $S$. Then 
    $$\Pr_{b_1, \ldots, b_n \in \mathbb{C}_d}[f(b_1, \ldots, b_n) = 0] \leq d/|S|.$$
\end{restatable}

We construct $\mathbb{C}_d$ in two stages. First, we construct a \emph{noncommutative}, nonassociative algebra $\mathbb{A}_d$ (see Section \ref{nonassnoncomm} for the precise definition) and show that $\mathbb{A}_d$ does not satisfy identities of degree $\leq d$. Note that this also gives us a noncommutative, nonassociative analogue of Theorem \ref{Comm-NonAssociative-SZ} (see Theorem \ref{NonAssociative-SZ} for exact statement). We prove this by showing that monomials in $\ncna$ can be {\em isolated} using substitutions from $(\mathbb{A}_d)^n$. Each monomial $m\in\ncna$ can be viewed as a rooted, ordered binary tree whose leaves are labeled by variables (see Figure \ref{fig:fig1} for examples). To each occurrence of a variable $x$ in $m$, we may associate a level which indicates the depth at which it appears in the monomial $m$ (when viewed as a tree). Noncommutativity of $\ncna$ also induces a {\em left to right} order in which the variables appear in $m$. We show that given the left to right order and the corresponding sequence of levels, we can fully reconstruct $m$ (Lemma \ref{levelSeq}). Using this lemma and a \emph{three dimensional version} of the set-multilinearization procedure introduced by Forbes and Shpilka \cite{FS13} in the context of PIT for noncommutative ABPs, we show that $\mathbb{A}_d$ does not satisfy identities of degree $\leq d$. The third dimension ``keeps track" of the level at which each leaf appears in a monomial. After this, we define the commutative algebra $\mathbb{C}_d$ as follows: the $\mathbb{C}_d$ product of $x, y$ is the \emph{anticommutator}\footnote{The anticommutator of $x,y$ with respect to a product operation $\cdot$ is defined as $x\cdot y + y\cdot x$.} of $x,y$ with respect to the $\mathbb{A}_d$ product. In $\mathbb{C}_d$, there is no unique left to right order of the variables that we can associate with a monomial. But there is a \emph{set} of orders that we can associate with each monomial. This set, together with the corresponding sequence of levels, determines the monomial uniquely. Using this, we show that $\mathbb{C}_d$ also does not satisfy identities of degree $\leq d$.

\subsubsection{White-box deterministic $\pit$ over \texorpdfstring{$\cna$}{Lg} (Section \ref{WhiteboxPIT})}

Next, we consider the white-box identity testing problem over $\cna$. Raz and Shpilka \cite{RS04} give a white-box linear algebraic algorithm for identity testing of noncommutative algebraic branching programs. Subsequently, their algorithm has been adapted to obtain PIT algorithms in various settings, for example, for Read-Once Algebraic Branching Programs (ROABPs) \cite{FS13}, noncommutative Unique Parse Tree Circuits \cite{Lagarde2019} and circuits over $\ncna$ \cite{ADMR17}. In this work, we show that an adaption of the Raz-Shpilka algorithm can be used to do $\pit$ for circuits over $\cna$:

\begin{restatable}{theorem}{DetPIT}
\label{PITAlg}
    Let $\Psi$ be a nonassociative arithmetic circuit of size $s$ computing an $n$ variate, degree $\leq d$ polynomial $f\in\cna$. Given $\Psi$ as input, we can check whether $f \equiv 0$ deterministically in time $\poly(s, n, d)$.
\end{restatable}

The main difference in the application of the Raz-Shpilka algorithm in our setting is that in all previous works (that we are aware of), if a monomial $m$ is generated at a product gate $g=g_1\times g_2$ in the circuit, then there is a \emph{unique way} it could have been generated: there exist monomials $m_1$, $m_2$ such that ${\sf coeff}_{m}(g)={\sf coeff}_{m_1}(g_1)\times{\sf coeff}_{m_2}(g_2)$. On the other hand since we are working in the commutative setting of $\cna$, there are two ways of generating $m$ at $g$. Either $m_1$ could be contributed by $g_1$ and $m_2$ by $g_2$, or $m_2$ could be contributed by $g_1$ and $m_1$ by $g_2$. We show (somewhat surprisingly) that the Raz-Shpilka algorithm, suitably modified, works even in this setting.

\subsubsection{Black-box deterministic nonassociative $\pit$ (Section \ref{deterministicBlackboxPIT})}

We consider next the question of derandomizing black-box $\pit$ for circuits over $\cna$. Towards this, we provide a \emph{hitting set} (consisting of elements of $(\mathbb{C}_d)^n$) for such circuits. A hitting set $H$ for a class $\mathcal{C}$ of circuits is a set of points such that for any non-zero circuit $\Psi\in \mathcal{C}$, there exists an $a\in  H$ such that $\Psi$ evaluated at $a$ is not $0$. 

\begin{restatable}{theorem}{hittingSet}
\label{commutative-hittingset}
    There exists a set $H_{n, s, d, \Delta}\subseteq (\mathbb{C}_d)^n$ of size $(nsd)^{O(\Delta)}$ of points in $(\mathbb{C}_d)^n$ such that for every nonassociative, commutative circuit $\Psi$ of size $\leq s$ and product depth $\leq \Delta$ computing a non-zero polynomial $f\in\cna$ of degree $\leq d$, there is a point in $H_{n, s, d, \Delta}$ at which $f$ is non-zero. Furthermore, we can compute $H_{n, s, d, \Delta}$ deterministically in time $(nsd)^{O(\Delta)}$.
\end{restatable}

Recall that the product depth of a circuit is the maximum number of product gates encountered on any leaf to root path in the circuit. Theorem \ref{commutative-hittingset} gives a non-trivial hitting set when the product depth $\Delta = o(d)$. In particular, when $\Delta$ is polylogarithmic, we obtain a quasipolynomial size hitting set.

The result of Kabanets and Impagliazzo \cite{KI04} shows that explicit lower bound results can give subexponential-time black-box $\pit$ algorithms over the usual commutative polynomial ring $\F[X]$. The main ingredients in their proof are the combinatorial design of Nisan and Wigderson \cite{NW94} and the factorization algorithm of Kaltofen \cite{Kal85}. Although we have strong and explicit lower bounds over the algebra $\cna$, it is unclear how to use them for $\pit$ algorithms. Note that such a connection is not known even over $\F\langle{X}\rangle$. 

We prove Theorem \ref{commutative-hittingset} in two stages. In the first stage, we reduce $\pit$ for circuits over $\cna$ to $\pit$ for \emph{unambiguous} circuits over $\mathbb{F}[Z]$ (where $Z$ is  a fresh set of variables and $\mathbb{F}[Z]$ is the usual polynomial ring) via a set-multilinearization argument. We say that a circuit $\Psi$ over is $\mathbb{F}[Z]$ is \emph{unambiguous} if for any monomial $m\in\mathbb{F}[Z]$, there exists a reduced parse tree\footnote{See \ref{reducedParseTree} for a definition.}\footnote{The term {\em unambiguous circuit} has been used in different contexts in earlier works \cite{AR16, LMP19}.} $T_m$ such that any reduced parse tree computing $m$ at any gate of $\Psi$ is isomorphic to $T_m$ as a labeled, rooted binary tree. These are the natural associative analogues of nonassociative circuits. We also observe that a similar reduction works over $\ncna$, which gives us an analogue of Theorem \ref{commutative-hittingset} over the algebra $\mathbb{A}_d$.

In the second stage, we suitably adapt the machinery of \emph{basis isolating weight assignments} developed by Agrawal et al.\ \cite{AGKS15} to construct hitting sets for unambiguous circuits. 

\begin{restatable}{theorem}{unambiguousHS}

\label{unambiguous-hittingset}

There exists a set $H_{s, n, d, \Delta}\subseteq \mathbb{F}^n$ such that for any unambiguous circuit $\Psi$ of size $s$ and product depth $\Delta$ computing a non-zero polynomial $f\in \mathbb{F}[z_1, \ldots, z_n]$ of degree $\leq d$, $f$ is non-zero on some point of $H_{s, n, d, \Delta}$. Furthermore, $|H_{s, n, d, \Delta}|=(nds)^{O(\Delta)}$ and $H_{s, n, d, \Delta}$ can be constructed in time $(nds)^{O(\Delta)}$. 
    
\end{restatable}

Informally, given a polynomial $f$ with coefficients coming from a vector space, a \emph{basis isolating weight assignment} for $f$ is a function from the underlying set of variables to $\mathbb{N}$ that isolates a minimum weight basis (among the coefficients) for the space spanned by the coefficients of $f$. Basis isolating weight assignments were used in \cite{AGKS15} to construct quasipolynomial size hitting-sets for ROABPs (these are commutative analogues of noncommutative ABPs). Subsequently, they were also used to construct hitting sets for set-multilinear Unique Parse Tree circuits (UPT circuits for short) \cite{ST18}. UPT set-multilinear circuits generalize ROABPs. In a UPT set-multilinear circuit, every parse tree (see \ref{parseTrees} for the definition) at the output has exactly the same shape as a rooted, ordered binary tree. In particular this implies that there is a universal parse tree shape at the root that induces a unique parse tree for each monomial. In unambiguous circuits, this is no longer true. In particular, there is a unique reduced parse tree \emph{for each monomial}, but two different monomials could have different parse tree shapes. Also, note that unambiguous circuits need not be set-multilinear. On the other hand, we require that for any monomial $m$, there is a unique parse tree computing it independent of the gate at which $m$ is being computed.

Suppose we have an unambiguous circuit of product depth $\Delta$. We construct a basis isolating weight assignment $w$ for it in multiple stages (as in \cite{AGKS15}). At each stage we handle monomials of increasing depths. The proof that $w$ is a basis isolating weight assignment involves isolation of a set $M_i$ of monomials for each depth $i\in [\Delta]$ such that the coefficient of every other monomial of depth $i$ is spanned by coefficients of $M_i$.  The construction of $M_i$ and the proof that its coefficients indeed span coefficients of other monomials is the main technical content of the proof and uses the fact that the circuit is unambiguous. Combining these $M_i$'s we get the isolated set $M$ of monomials. Identity testing follows from the construction of a basis isolating weight assignment.


The result of Valiant, Skyum, Berkowitz and Rackoff \cite{VSBR83} shows that arithmetic circuits can be depth reduced to depth polylogarithmic in the size and degree of the original circuit while incurring only a polynomial blowup in size. Unfortunately, we do not know if such a depth reduction is possible while preserving unambiguousness.

\subsubsection*{Organization} The paper is organized as follows. In Section \ref{sec:prelim}, we provide the necessary background. Section \ref{Randomized} contains the randomized polynomial-time $\pit$ algorithms over nonassociative algebras. The deterministic $\pit$ algorithms (white-box and black-box) are presented in Section \ref{sec:det-algo}. We state a few questions for further research in Section \ref{sec:discuss}. 
\section{Preliminaries}\label{sec:prelim}

\begin{defn}[Algebra over a field]
    Let $\mathbb{F}$ be a field. An \textit{algebra} $\mathbb{A}$ over $\mathbb{F}$ is an $\mathbb{F}$-vector space together with a product operation on the elements of the vector space that is \textit{bilinear}. The \textit{dimension} of $\mathbb{A}$ is defined to be the dimension of the underlying vector space. In particular, if the underlying vector space is finite (say $n$) dimensional and identified with $\mathbb{F}^n$ after choice of a basis, an algebra $\mathbb{A}$ is uniquely defined by $n$ matrices $L_1, \ldots, L_n \in \mathbb{F}^{n \times n}$ as follows: for $\mathbf{x}, \mathbf{y} \in \mathbb{F}^n$, their $\mathbb{A}$-product is precisely $$\mathbf{x} \cdot \mathbf{y} = (x_1, \ldots, x_n)\cdot(y_1, \ldots, y_n) = \left(\mathbf{x}^{T}L_1\mathbf{y}, \ldots, \mathbf{x}^{T}L_n \mathbf{y}\right).$$ An algebra $\mathbb{A}$ is called \emph{unital} if it contains a multiplicative identity $i$ such that $\forall x\in\mathbb{A}, x\cdot i=i\cdot x = x$. We define the following four different polynomial algebras, depending on the relations the variables satisfy:
    \begin{itemize}
    \item[(1)] $\mathbb{F}_{A,C}[X]$: This is the polynomial ring $\mathbb{F}[X]$. The product operation is both commutative and associative.
    \item[(2)] $\mathbb{F}_{A,\widebar{C}}[X]$ is the noncommutative polynomial ring $\mathbb{F}\langle X \rangle$. The product operation is noncommutative but associative.
    \item[(3)] $\mathbb{F}_{\bar{A},C}[X]$ is the $\mathbb{F}$-vector space generated by commutative, nonassociative monomials in the variables $X$. This vector space becomes an $\mathbb{F}$-algebra with the commutative, nonassociative product of monomials extended to all of $\cna$ by bilinearity.
    \item[(4)] $\mathbb{F}_{\bar{A},\bar{C}}[X]$ is the $\mathbb{F}$-vector space generated by noncommutative, nonassociative monomials in the variables $X$. This vector space becomes an $\mathbb{F}$-algebra with the noncommutative, nonassociative product of monomials extended to all elements of $\mathbb{F}_{\bar{A},\bar{C}}[X]$ by bilinearity. 

    \end{itemize}
\end{defn}

\begin{defn}[Polynomial Identities]
    A \textit{Polynomial Identity} (PI for short) for an 
    algebra $\mathbb{A}$ is a polynomial $f(x_1, \ldots, x_n)$ in a 
    set of variables $\{x_1, \ldots, x_n\}$ such that for all $A_1, \ldots, A_n \in \mathbb{A}$, $f(A_1, \ldots, A_n) = 0$ where the multiplication is according to the product operation in $\mathbb{A}$. An algebra that satisfies nontrivial identities is called a \emph{PI-algebra}.
\end{defn}

The study of polynomial identities is a classical and very rich subject in mathematics. For a comprehensive details, see \cite{GZ05}.

\begin{defn}[Arithmetic Circuit]
An \textit{Arithmetic Circuit} $\Psi$ over a field $\mathbb{F}$ is a directed acyclic graph whose leaves (called input gates) are labeled by either variables (say $X = \{x_1, \ldots, x_n\}$) or field elements and whose internal vertices (called gates) are labeled by  either a sum ($+$) or a product $(\times)$. In our case the product operation will often be \textit{nonassociative}, and we will assume that the fan-in of each product gate is $2$. If in addition we are working in $\mathbb{F}_
{\bar{A}, \bar{C}} [X]$, the product will also be noncommutative: each product gate will have designated left and and right child.  Each gate in a circuit naturally computes a polynomial. The circuit $\Psi$ has a designated output gate and $\Psi$ is said to compute the polynomial computed at the output gate. The {\em size} of a circuit is the number of gates in it and the {\em depth} of a circuit is the length of the longest leaf-to-root path. 
\end{defn}

Next we define the concept of a {\em parse tree}, that depicts the generation of a particular monomial in the circuit.

\begin{defn}[Parse trees]\label{parseTrees}
Let $X = \{x_1, \ldots, x_n\}$ be a set of variables, $\mathbb{F}$ be a field and $\Psi$ be an arithmetic circuit computing a polynomial $f\in\mathbb{F}[X]$. The set of {\em parse trees} for $\Psi$ will be defined by induction on the size of $\Psi$ as follows:
\begin{itemize}
    \item If $\Psi$ is just a leaf labeled by either a variable or a constant, then it has only one parse tree, itself. 
    \item If the root $g$ of $\Psi$ is a sum gate with subcircuits $\Psi_1$ and $\Psi_2$ as children, the set of parse trees for $\Psi$ is the set of all trees obtained by taking the root $g$, and attaching to it a parse tree of either $\Psi_1$ or $\Psi_2$.
    
    \item If the root $g$ of $\Psi$ is a product gate with subcircuits $\Psi_1$ and $\Psi_2$, we define the set of parse trees for $\Psi$ to be the set of all trees $T$ obtained by taking a parse tree $T_1$ for $\Psi_1$, a parse tree $T_2$ for a disjoint copy of $\Psi_2$ and making $T_1, T_2$ the children of $g$. 
\end{itemize}
Note that each parse tree $T$ for $\Psi$  computes a monomial (with coefficient).
\end{defn}

\begin{defn}[Reduced Parse Trees]\label{reducedParseTree}
    From each parse tree $T$ of a circuit $\Psi$, we may obtain a \emph{reduced parse tree} $T'$ by short-circuiting the sum gates, removing leaves labeled by constants and restructuring the tree in the natural way. $T'$ is a full binary tree all of whose leaves are labeled by a variable and all of whose gates are product gates. $T'$ captures the multiplicative structure of $T$. The \emph{set of reduced parse trees} for $\Psi$ is defined as $\{T'\mid T\text{ is a parse tree for }\Psi\}$. See Figure \ref{parseTree} for an illustrative example.
\end{defn}

\begin{figure}
\begin{subfigure}[b]{0.37\textwidth}
\centering
\begin{tikzpicture}[edge from child/.style={draw,-latex}, level 1/.style={sibling distance=35mm}, level 2/.style={sibling distance=20mm}, level 3/.style={sibling distance=10mm},scale=0.75]
  \node [circle,draw]  {$\times$}
  child {node [circle,draw] {$+$}
   child {node [circle,draw] {$\times$} edge from parent
   child {node {$x_1$}}
   child {node {$x_1$}
   }}
   child {node [circle,draw] {$+$} edge from parent
   child {node {$x_2$}}
   child {node {$3$}}
   }
   }
    child {node [circle,draw] {$+$}
    child {node [circle, draw] {$\times$} edge from parent
      child {node {$x_3$}}
   child {node {$x_4$}}   
    }
   child {node [circle,draw] {$\times$} edge from parent
   child {node {$6$}}
   child {node {$x_5$}} }
};
\end{tikzpicture}
\caption{A circuit $\Psi$}
\end{subfigure}
\begin{subfigure}[b]{0.35\textwidth}
\centering
\begin{tikzpicture}[edge from child/.style={draw,-latex}, level 1/.style={sibling distance=35mm}, level 2/.style={sibling distance=20mm}, level 3/.style={sibling distance=10mm},scale=0.75]
  \node [circle,draw]  {$\times$}
  child {node [circle,draw] {$+$}
   child {node [circle,draw] {$\times$} edge from parent
   child {node {$x_1$}}
   child {node {$x_1$}}}
   child [missing]
   }
    child {node [circle,draw] {$+$}
    child [missing]
   child {node [circle,draw] {$\times$} edge from parent
    child {node {$6$}}
   child {node {$x_5$}} }
   };
\end{tikzpicture}
\caption{A parse tree $T$ for $\Psi$}
\end{subfigure}
\begin{subfigure}[b]{0.25\textwidth}
\centering
\begin{tikzpicture}[edge from child/.style={draw,-latex}, level 1/.style={sibling distance=15mm}, level 2/.style={sibling distance=10mm}, level 3/.style={sibling distance=10mm},scale=0.75]
  \node [circle,draw]  {$\times$}
   child {node [circle,draw] {$\times$} edge from parent
   child {node {$x_1$}}
   child {node {$x_1$}}}
   child {node {$x_5$} 
   };
\end{tikzpicture}
\caption{Reduced parse tree $T'$}
\label{parseTree}
\end{subfigure}
\end{figure}

We also recall the following standard result over $\mathbb{F}_{A,C}[X]$: 
\begin{lemma}[Polynomial Identity Lemma \cite{Zip79, Sch80, DL78}]
\label{lem:SZ}
    Suppose $f(x) \in \mathbb{F}[X]$ is an $n$-variate polynomial of degree $d$, and let $S \subseteq F$ be a finite set of size strictly larger than $d$. Then
$f(\bar{a})\neq 0$ for at least $\left( 1-\frac{d}{|S|}\right)$ fraction of $\bar{a}$'s in $S^n$.
\end{lemma}

Over the algebra $\mathbb{F}_{A,\bar{C}}[X]$, the following  is a well-known result of Amitsur and Levitzki \cite{AL50}. 

\begin{theorem}[Amitsur-Levitzki Theorem \cite{AL50}]\label{thm:amitsur-Levitzki}
Over any field $\mathbb{F}$, the matrix algebra $\mathbb{F}^{k\times k}$ satisfies no PI of (total)
degree less than $2k$, and satisfies exactly one (up to constant
factor) PI of degree $2k$.  
\end{theorem}

\subsection{A structural lemma about nonassociative monomials}\label{prelimLemma}

\subsubsection{Over the algebra \texorpdfstring{$\mathbb{F}_{\bar{A}, \bar{C}}[X]$}{Lg}} Let $\mathbb{F}$ be a field and $X=\{x_1, \ldots,x_n\} $ be a set of variables. There is a natural correspondence between rooted binary trees with leaves labeled by elements of $X$ and monomials in $\ncna$. The nonassociativity introduces a unique product structure which may be interpreted as a binary tree: Let $m\in\ncna$ be a monomial. Then there is a unique rooted binary tree $T_m$ whose the leaves (labeled by elements of $X$) represent the variables, and whose internal nodes compute the product their two children. The root computes the monomial $m$. For instance, the binary trees in Figure 
\ref{fig:fig1}(a) and \ref{fig:fig1}(b) compute monomials $m=((x_{i_1} x_{i_2})x_{i_3})$ and $m'=(x_{i_1} (x_{i_2}x_{i_3}))$ respectively. Noncommutativity implies that each internal node has designated left and right children, and swapping this order changes the monomial. 

Suppose $m$ has degree $d$. Noncommutativity gives us a unique string $\sigma_{m}=(i_1, \ldots, i_d)\in[n]^d$ which is the unique {\em left to right order} $x_{i_1},x_{i_2}\ldots x_{i_d}$ in which the variables appear in $m$. We will think of $\sigma_m:[d]\rightarrow[n]$ as a function defined as $\sigma_m(j) = i_j$ for all $j\in[d]$. Note that $\sigma_m$ does not uniquely define a monomial. For instance, for the monomials $m, m'\in\ncna$ shown in Figure \ref{fig:fig1}, $\sigma_m=\sigma_{m'}$ but $m\neq m'$. Given a monomial $m\in\ncna$ as a binary tree $T_m$, we assign \textit{level numbers} to the nodes of $T_m$. Define the level of a node $v$ in $T_m$ to be $1+d(\text{root}, v)$, where $d(\cdot, \cdot)$ is the {\em distance} function. That is, the root is at level $1$, the children of the root are at level $2$ and so on. For each $j\in[d]$, let $l^m_j$ denote the level at which the $j^{th}$ variable (in the left to right order) appears in $m$. 

For monomials $m=((x_{i_1} x_{i_2})x_{i_3})$ and $m'=(x_{i_1} (x_{i_2}x_{i_3}))$ shown in Figure \ref{fig:fig1}, $l_1^m = 3 , l_2^m=3, l_3^m=2$ and  $l_1^{m'} = 2, l_2^{m'} = 3, l_3^{m'}=3$. In general, the variable order $\sigma_m$ and the level numbers $(l_1^m, \ldots, l_d^m)$ \emph{together} uniquely determine the monomial $m$. 

\begin{figure}
\begin{subfigure}[b]{0.32\textwidth}
   \centering
  \begin{tikzpicture}
\node[circle,draw](z){$\times$}
  child{node[circle,draw]{ $\times$}  child{node[]{$x_{i_1}$}} child{node[]{$x_{i_2}$}}}
  child{
    node[]{$x_{i_3}$}};
\end{tikzpicture}
\caption{$m=((x_{i_1} x_{i_2})x_{i_3})$}
\end{subfigure}
\begin{subfigure}[b]{0.32\textwidth}
\centering
  \begin{tikzpicture}
    \node[circle,draw](z){$\times$}
  child{node[]{$x_{i_1}$}} 
  child{node[circle,draw]{$\times$} } 
  {
  child{node[]{$x_{i_2}$}} 
  child{node[]{$x_{i_3}$}}
  }
  ;
\end{tikzpicture}
\caption{$m' = (x_{i_1} (x_{i_2}x_{i_3}))$}
\end{subfigure}
\begin{subfigure}[b]{0.32\textwidth}
\centering
    \begin{tikzpicture}
\node[circle,draw](z){$\times$}
  child{node[circle,draw]{ $\times$}  child{node[]{$x_{i_2}$}} child{node[]{$x_{i_1}$}}}
  child{
    node[]{$x_{i_3}$}};
\end{tikzpicture}
\caption{$m''=((x_{i_2} x_{i_1})x_{i_3})$}
\end{subfigure}
\caption{Examples of nonassociative, noncommutative monomials} 
\label{fig:fig1}
\end{figure}
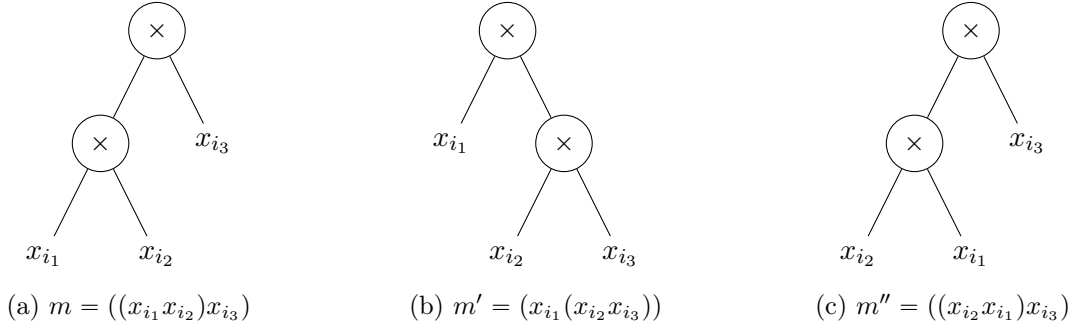

We state this formally in the following lemma:


\begin{lemma}
\label{levelSeq}
    Let $m,m'$ be distinct monomials in $\ncna$. Let ${\sf deg}(m) = d$ and ${\sf deg}(m') = d'$. Then the tuples $(\sigma_m, l^m_1, \ldots, l^m_d)$ and $(\sigma_{m'}, l^{m'}_1, \ldots, l^{m'}_{d'})$ are distinct.
\end{lemma}

\begin{proof}
    We assume that $\sigma_m = \sigma_{m'}$, for otherwise the statement is evidently true. In particular, assume that the degrees of $m$ and $m'$ are both equal to $d$. Next, we simply observe that it is possible to iteratively reconstruct $m$ uniquely, given the sequence $l^m_1, \ldots, l^m_d$ of levels. Start with just the root. To it, add a path of length $l_{1}^m-1$ consisting entirely of left edges (every node is a left child of it's parent). Label the ultimate node (leaf) on the path by $x_{\sigma_m(1)}$. Now suppose we already have a leaf for the $i$-th variable, $x_{\sigma_m(i)}$. To generate the leaf labeled by $x_{\sigma_m(i+1)}$, find the ancestor $v$ of $x_{\sigma_m(i)}$ closest to $x_{\sigma_m(i)}$ that only has a left child. Let the \textit{level} of $v$ be $l$. Add a right $v'$ to child to $v$. If $v'$ is at level $l_{\sigma_m(i+1)}$, label it by $x_{\sigma_m(i+1)}$. Otherwise, add to it a path of length $l_{\sigma_m(i+1)}-l-1$ consisting only of left edges and label the ultimate node on this path by $x_{\sigma_m(i)}$. This process recovers $m$. The key property of trees representing nonassociative monomials that is used in the procedure above is that they are \textit{full} binary trees, that is, each node either has exactly two children or none at all. This justifies our choice of going back to the the closest ancestor $v$ of $x_{\sigma_m(i)}$ that only has a left child and no right child.
\end{proof}

\subsubsection{Over the algebra \texorpdfstring{$\mathbb{F}_{\bar{A}, {C}}[X]$}{Lg}} 
Next, we consider the case of monomials in $\cna$. 
One may partition the set of all monomials in $\ncna$ into equivalence classes under commutativity: For all monomials $m, m'\in \ncna$, $m\sim m'$ if and only if $m$ and $m'$ are equal up commutativity. For example the monomials shown in Figure \ref{fig:fig1}(a) and Figure \ref{fig:fig1}(b) belong to different equivalence classes. On the other hand, monomials in Figure \ref{fig:fig1}(a) and Figure \ref{fig:fig1}(c) belong to the same equivalence class.   

Each equivalence class may be identified with a monomial in $\cna$. Looking at it from the other direction, suppose $m, m'$ are monomials in $\cna$ and suppose $M_m$, $M_{m'}$ are the equivalence classes of monomials in $\ncna$ they represent. Since $\sim$ as defined above is an equivalence relation, we have the following simple observation:

\begin{observation}
\label{commutativeObs}
    For any two distinct monomials $m,m'\in\cna$ we have $M_{m}\cap M_{m'} = \emptyset$. 
\end{observation}

\section{Randomized Black-box \texorpdfstring{$\pit$}{Lg} for Nonassociative circuits}\label{Randomized}
Let $\mathbb{F}$ be a field and let $X=\{x_1, \ldots,x_n\}$. In this section, we work with both the nonassociative, noncommutative polynomial algebra $\ncna$ as well as the commutative $\cna$. 
In subsection \ref{nonassnoncomm}, we give a randomized, polynomial time black-box algorithms to test whether a nonassociative, noncommutative circuit $\Phi$ over $\mathbb{F}$ computes the identically zero polynomial.  The algorithm assumes that we have access to evaluations of $\Phi$ on a suitable $k$ dimensional $\mathbb{F}$-algebra $A$, and that the cost of one $A$-query is $\poly({\sf size}(\Phi), k)$. That is, the cost is polynomial in the size of the circuit and the dimension of the algebra. To the best of our knowledge, this gives the first Amitsur-Levitzki type theorem \cite{AL50} over nonassociative polynomial algebras.


\subsection{Nonassociative, Noncommutative Randomized Black-box \texorpdfstring{$\pit$}{Lg}}
\label{nonassnoncomm}

The key idea is to construct a noncommutative, nonassociative $\mathbb{F}$-algebra and show that it does not satisfy low degree polynomial identities. This will imply that a random non-zero substitution from this algebra will make a non-zero circuit $\Phi$ evaluate to something non-zero, with high probability. 

We will query noncommutative, nonassociative circuits computing a polynomial $f\in\ncna$ of degree $\leq d$ on a particular $d(d+1)^2+1$ dimensional algebra $\mathbb{A}_d$ that we now describe. First, we construct an algebra $\mathbb{A'}_d$ of dimension $d(d+1)^2$ and then construct the desired algebra $\mathbb{A}_d$ by adjoining an identity element to $\mathbb{A'}_d$. Additively, $\mathbb{A'}_d$ is an $\mathbb{F}$-vector space of dimension $d(d+1)^2$. We will think of an element of $\mathbb{A'}_d$ as a set of $d$ matrices each of dimension $(d+1)\times (d+1)$:

\begin{figure}
    \centering
    \includegraphics[width=0.5\linewidth]{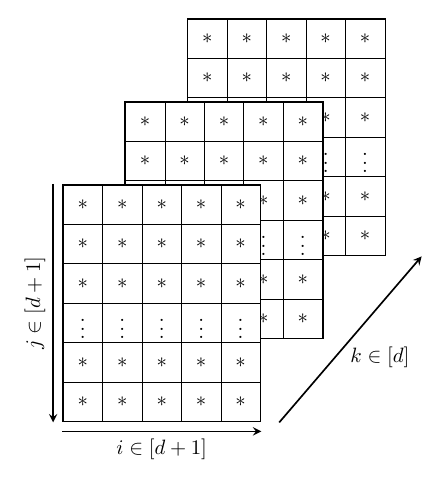}
    \caption{$3$-dimensional view of an element of $\mathbb{A}'_d$}
    \label{fig:algebra-element}
\end{figure}


Let $x, y \in \mathbb{F}^{d(d + 1)^2}$. We index $x, y$ as $x[i, j, k]$ and $y[i, j, k]$ for $1 \leq i, j \leq d + 1$ and $1 \leq k \leq d$. Here, $k$ can be thought of as indexing the set of $d$ matrices and for a fixed $k$, $i$ and $j$ index respectively the rows and columns of the $k$-th matrix. Next, we define the bilinear $\mathbb{A}_d'$-product of $x, y$ as follows: $z \triangleq x \circ y$ such that 

\[ z[i, j, k] = \begin{cases} 
        0 & k = d\\
        \displaystyle\sum_{l = 1}^{d + 1}x[i, l, k + 1]y[l, j, k + 1] & 1 \leq k \leq d-1 \\
   \end{cases}\]
Clearly, $\mathbb{A'}_d$ is an $\mathbb{F}$-algebra. From the proof of Lemma \ref{lowerBound}, it will be evident that $\mathbb{A'}_d$ is a nonassociative algebra. 

We require the algebra $\mathbb{A'}_d$ to be unital to make sense of $\mathbb{A'}_d$-evaluations of polynomials with a non-zero constant term: For any $f \in \ncna$, $f(0, \ldots, 0) = c\cdot \mathbf{1}$ where $f$ has constant term $c\in \mathbb{F}$ and $\mathbf{1}$ is the identity element of $\mathbb{A'}_d$. However, $\mathbb{A'}_d$ is non-unital. To construct the unital algebra $\mathbb{A}_d$ (from $\mathbb{A'}_d$), we use a standard procedure of adjoining an identity to an algebra: 
\paragraph{The Algebra $\mathbf{\mathbb{A}_d}$:} Define the algebra $\mathbb{A}_d$ to be the vector space $\mathbb{F}^{d(d + 1)^2+1} = \{(a, \alpha)\mid a \in \mathbb{A'}_d, \alpha \in \mathbb{F}\}$ together with the bilinear product $\cdot$ defined as follows:
$$(a_1, \alpha_1)\cdot(a_2,\alpha_2) = (a_1\circ a_2+\alpha_1 a_2 + \alpha_2 a_1, \alpha_1\alpha_2)$$
We will refer to this newly added index as the last index. It is easy to verify that $\mathbb{A}_d$ is indeed a nonassociative algebra with $(\mathbf{0}, 1)$ being the multiplicative identity, and that $\mathbb{A'}_d$ is isomorphic to the sub-algebra $\{(a, 0)\mid a \in \mathbb{A'}_d\}$ of $\mathbb{A}_d$. \\

Now, we show that $\mathbb{A}_d$ cannot have polynomial identities of degree $\leq d$. 


\begin{lemma}
\label{lowerBound}
    Let $f\in \ncna$ be a non-zero polynomial of total degree $\leq d$. Then $f$ is not a PI for $\mathbb{A}_d$.
\end{lemma}

\begin{proof}
It suffices to consider polynomials that do not have a constant term, since a polynomial with a non zero constant term evaluates to a non-zero value at the all zeroes input. Also, it suffices to prove Lemma \ref{lowerBound} for the algebra $\mathbb{A'}_d$ instead of $\mathbb{A}_d$, since $\mathbb{A'}_d$ is a subalgebra of $A_d$. In order to prove the lemma, we reduce the problem to the associative, commutative setting. 

For each $x_i$, we introduce an \textit{associative, commutative} set of variables $\{z_{i, j, k}\mid 1 \leq j, k \leq d\}$. For convenience, denote the vector $(z_{i, j, k})_{j, k \in[d]}$ by $\mathbf{z_i}$. 
Extend the ground field $\mathbb{F}$ to the function field $\mathbb{F}'=\mathbb{F}(\mathbf{z}_1, \ldots, \mathbf{z}_n)$. We define the algebra $\mathbb{A'}_d$ over the field $\mathbb{F}'$ as described earlier. Eventually, the $\mathbf{z}$ variables will be fixed suitably from the base field $\mathbb{F}$. 

Next, consider the evaluation map $\Phi:\ncna\rightarrow \mathbb{A'}_d$
that sends $x_i$ to $Z_i$ where for each $1 \leq j, k \leq d$, the $(j, j + 1,k)$-th entry of $Z_i$ is $z_{i, j, k}$ and the rest of the entries are zero. For an illustration, we describe $Z_1$ explicitly in Figure \ref{fig:Z_1}.

\begin{figure}
    \centering    \includegraphics[width=0.5\linewidth]{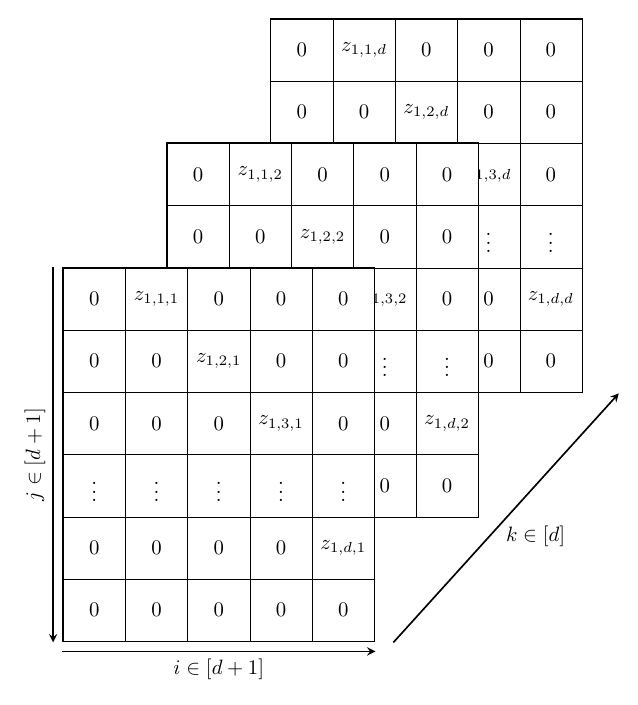}
    \caption{$Z_1 = \Phi(x_1)$ visualized as an element of $\mathbb{A}'_d$ 
    }
    \label{fig:Z_1}
\end{figure}

Let us look at the image of a monomial $m\in \ncna$ (of degree at most $d$) under $\Phi$. Let $d'\leq d$ be the degree of $m$. We interpret $m$ as a binary tree with leaves labeled by variables. Since we are in the noncommutative setting, there is a unique function $\sigma_m:[d']\rightarrow [n]$ that describes the left-to-right order in which the variables appear in $m$. We will \textit{level} the nodes (including leaves) of $m$ as described in Section \ref{prelimLemma}. Let $l^m_t$ denote the \textit{level} at which the $t$-th variable in $\sigma_m$ (i.e., variable $x_{\sigma_m(t)}$) appears in $m$. Let the \textit{depth} of $m$ (i.e., ${\sf max}_i\{l_i\}$) be $l$. 

    \begin{claim}
    \label{claim:polynomialEntry}
        Let $m$ be as above. For each $k_1\in [d-d'+1]$ and each $k_2\in [d - l + 1]$, the $(k_1, k_1 + d', k_2)$-th entry of $\Phi(m)$ is the monomial $\prod_{t = 1}^{d'} z_{\sigma_m(t), t + k_1 - 1, l^m_{t} + k_2 - 1}$, and every other entry is zero.
    \end{claim}

    \begin{proof}
        We prove this by induction on the degree $d'$ of $m$. For ${\sf deg}(m) = 1$, the claim follows from the definition of $Z_i$'s. Now suppose $d'>1$ and $m$ is uniquely written as $m_1m_2$ such that ${\sf deg}(m_1) = d_1$, ${\sf deg}(m_2) = d_2$ and $d_1 + d_2 = d'$. Then,
        \begin{align*}
            \Phi(m)[k_1, k_1 + d', k_2] = \sum_{i = 1}^{d+1}\Phi(m_1)[k_1, i, k_2+1]\Phi(m_2)[i, k_1+d', k_2+1]
        \end{align*}
        By induction, we see that exactly one term in this sum is non zero, the one corresponding to $i=k_1+d_1$ and the sum is therefore equal to 
        \begin{equation}\label{eq1}
            \left(\prod_{t = 1}^{d_1}z_{\sigma_{m_1}(t), t + k_1 - 1, l^{m_1}_{t}+k_2}\right)\left(\prod_{t = 1}^{d_2}z_{\sigma_{m_2}(t), t + d_1 +k_1 - 1, l^{m_2}_{t}+k_2}\right)
        \end{equation}

        Notice that

        \begin{equation*}
        \sigma_m(t)=\begin{cases}
                \sigma_{m_1}(t) \quad 1 \leq t \leq d_1 \\
                \sigma_{m_2}(t-d_1) \quad d_1+1\leq t \leq d' 
            \end{cases}
        \end{equation*}

        and that 

        \begin{equation*}
        l^m_t=\begin{cases}
                l^{m_1}_t+1 \quad 1 \leq t \leq d_1 \\
                l^{m_2}_{t-d_1}+1 \quad d_1+1\leq t \leq d' 
            \end{cases}
        \end{equation*}

        Using these observations, we find that (\ref{eq1}) is exactly equal to $\prod_{t = 1}^{d'} z_{\sigma_m(t), t + k_1 - 1, l^m_{\sigma_m(t)} + k_2 - 1}$. \\

        Now let us look at $\Phi(m)[i_1, i_2, i_3]$ such that $i_2 \neq i_1 + d'$ and $i_3 \in [d - 1]$ (if $i_3 = d$, $\Phi(m)[i_1, i_2, i_3] = 0$ by definition of the $\mathbb{A}_d$ product). 

        $$\Phi(m)[i_1,i_2,i_3] = \sum_{i = 1}^{d+1}\Phi(m_1)[i_1, i, i_3+1]\Phi(m_2)[i, i_2, i_3+1]$$

        Using the induction hypothesis, we see that each summand in the sum above is actually zero, and therefore so is $\Phi(m)[i_1,i_2,i_3]$. 

        Let us also look at $\Phi(m)[i_1, i_2, i_3]$ with $i_3>d-l+1$. If $d'=2$ then all such entries of $\Phi(m)$ are easily seen to be zero. Now suppose $d'>2$ and assume, without loss of generality, that the depth of $m_1$ is $l-1$. Then by induction, all entries $\Phi(m_1)[i_1, i_2, i_3]$ of $\Phi(m_1)$ with $i_3>d-l$ are zero and therefore so are all the entries $\Phi(m)[i_1, i_2, i_3]$ of $\Phi(m)$ with $i_3>d-l+1$.
        
        This concludes the proof of Claim \ref{claim:polynomialEntry}.
    \end{proof}
    Now, by setting $k_1, k_2=1$ in the statement of Claim \ref{claim:polynomialEntry} we see that for a monomial $m$ of degree $d'$, the $(1, d'+1, 1)$-th entry of $\Phi(m)$ is  $\prod_{t = 1}^{d'} z_{\sigma_m(t), t, l_{t}}$. For any monomial of degree $\neq d'$, this entry is zero. This observation combined with Lemma \ref{levelSeq} gives us Lemma \ref{lowerBound}. 
\end{proof}

Using Lemma \ref{lowerBound}, we exhibit a randomized black-box identity testing algorithm for nonassociative, noncommutative circuits.
\begin{theorem}\label{NonAssociative-SZ}
    Let $\mathbb{F}$ be a field with $|\mathbb{F}|>d$, and $S \subset \mathbb{F}$. Let $f \in \ncna$ be a non-zero polynomial of degree $\leq d$ given as a black-box with query access to evaluations of $f$ on elements of $\mathbb{A}_d$. 
    Let $S\subseteq \mathbb{F}$ with $|S|>d$. Sample $b_1, \ldots, b_n\in \mathbb{A}_{d}$ as follows: Pick each of the $d(d+1)^2+1$ entries of each of the $b_i$'s uniformly and independently from $S$. Then 
    $$\Pr_{b_1, \ldots, b_n \in \mathbb{A}_d}[f(b_1, \ldots, b_n) = 0] \leq d/|S|.$$
\end{theorem}

\begin{proof}
From Lemma \ref{lowerBound}, it follows that $f$ is not a PI for $\mathbb{A}_d$. By slight abuse of notation, define $\mathbb{A}_d$ over the field $\mathbb{F}(Y)$ where $Y$ is a commutative, associative set of variables and $|Y|=nd(d+1)^2+1$. Replace each $x_i$ by $Y_i\in \mathbb{A}_d$ defined as follows: $Y_i$ has dimension $d(d+1)^2+1$ and each entry of each $Y_i$ is a fresh variable from $Y$. Since $f$ is not a PI for $\mathbb{A}_d$, $f(Y_1, \ldots, Y_n)\nequiv 0$ and so at least one entry of $f(Y_1, \ldots, Y_n)$ is a non-zero polynomial in $\mathbb{F}[Y]$.
By the Polynomial Identity Lemma (Lemma \ref{lem:SZ}), under a random $S$-substitution, the probability that that entry of $f(Y_1, \ldots, Y_n)$ evaluates to zero is $\leq d/|S|$.
\end{proof}

Theorem \ref{NonAssociative-SZ} can be thought of as a version of Theorem \ref{thm:amitsur-Levitzki} over 
$\mathbb{F}_{\bar{A}, \bar{C}}$. 
It immediately gives us the desired black-box PIT algorithm.

\subsection{Nonassociative, Commutative Randomized Black-box $\pit$}
\label{blackbox-commutative}

Next, we construct a nonassociative, \emph{commutative} algebra $\mathbb{C}_d$ that does not satisfy low degree identities in $\cna$. $\mathbb{C}_d$ is constructed using the algebra $\mathbb{A}_d$ from the previous section (Section \ref{nonassnoncomm}). Recall that $\cdot$ denotes the $\mathbb{A}_d$-product. 

\paragraph{The Algebra $\mathbb{C}_d$:} 
$\mathbb{C}_d$ is isomorphic to $\mathbb{A}_d$ as a vector space, and the $\mathbb{C}_d$ product of $a, b$ (denoted by $\odot$) is simply the \emph{anticommutator} of $a, b$ with respect to the $\mathbb{A}_d$ product $``\cdot"$. That is, $a \odot b = a \cdot b + b \cdot a$. Let $\mathbb{C'}_d$ denote the sub-algebra of $\mathbb{C}_d$ obtained by setting the last index to $0$. It is easily verified that $\mathbb{C}'_d$ is isomorphic to the algebra whose product is the anticommutator with respect to the $\mathbb{A}'_d$ product $\circ$. For convenience, in the sequel, we will drop the last index of elements of $\mathbb{C}'_d$.

\nonassocAL*

\begin{proof}
    
    The proof of Lemma \ref{lowerBound2} is similar to the proof of Lemma \ref{lowerBound}. As before, we note that it suffices to prove the lemma for constant free polynomials and that it suffices to prove the claim for $\mathbb{C'}_d$ instead of $\mathbb{C}_d$, because any identity of $\mathbb{C}_d$ is also an identity of $\mathbb{C'}_d$.
    

    For each $x_i$, we introduce the same \textit{associative, commutative} set of variables $\{z_{i, j, k}\mid 1 \leq j, k \leq d\}$ as before. Denote the vector $(z_{i, j, k})_{j, k \in[d]}$ by $\mathbf{z_i}$. As in the proof of Lemma \ref{lowerBound}, we consider the extended field $\mathbb{F}'=\mathbb{F}(\mathbf{z}_1, \ldots, \mathbf{z}_n)$ and define $\mathbb{C'}_d$ over $\mathbb{F}'$.  
    We consider the evaluation map $\Phi:\cna\rightarrow \mathbb{C'}_d$ that sends $x_i$ to $Z_i$ where for each $1 \leq j, k \leq d$, the $(j, j + 1,k)$-th entry of $Z_i$ is $z_{i, j, k}$ and the rest of the entries are zero. Again, we would like to inspect the image of a monomial $m\in \cna$ of degree $d'\leq d$ under $\Phi$. As before, we interpret $m$ as a binary tree. Without loss of generality, let the variables appearing in $m$ be $x_1, \ldots, x_{d'}$. Unlike in the noncommutative case, there is no unique left to right order of variables that one can associate with the monomial $m$. 
    
    There is, however, a \textit{set} of orders that one can associate with $m$: For each internal node in the tree representing $m$, arbitrarily designate one of the children to be the left child and the other to be the right child. This procedure induces an order $\sigma$ on the variables of $m$. Furthermore, every distinct way of designating left and right children at the internal nodes of $m$ induces a unique order. Consider the union of all these orders and denote this set by $\Sigma_m$. Each $\sigma\in\Sigma_m$ corresponds to a unique monomial from $\ncna$ in the equivalence class $M_m$ corresponding to the monomial $m$ (see Section \ref{prelimLemma}, the commutative case). Also, for a fixed $\sigma\in\Sigma_m$, let $l^{m, \sigma}_t$ denote the \textit{level} at which the $t$-th variable in the order $\sigma$ (i.e., variable $x_{\sigma(t)}$) appears in $m$. Let the \textit{depth} of $m$ (i.e., ${\sf max}_{i}\{l_i\}$) be $l$.

    \begin{claim}\label{polynomialEntry2}
        Let $m$ be as above. For each $k_1\in [d-d'+1]$ and each $k_2\in [d - l + 1]$, the $(k_1, k_1 + d', k_2)$-th entry of $\Phi(m)$ is the polynomial $\displaystyle\sum_{\sigma \in \Sigma_m} \prod_{t = 1}^{d'} z_{\sigma(t), t + k_1 - 1, l^{m, \sigma}_{t} + k_2 - 1}$, and every other entry is zero.
    \end{claim}

    \begin{proof}
    We prove this by induction on the degree $d'$ of $m$. For ${\sf deg}(m) = 1$, $|\Sigma_m| = 1$ and the claim follows from the definition of $Z_i$'s. Now suppose $d'>1$ and $m$ is uniquely written as $m_1m_2$ such that ${\sf deg}(m_1) = d_1$, ${\sf deg}(m_2) = d_2$ and $d_1 + d_2 = d'$. Then,

        \begin{align*}
            \Phi(m)[k_1, k_1 + d', k_2] = &\sum_{i = 1}^{d+1}\Phi(m_1)[k_1, i, k_2+1]\Phi(m_2)[i, k_1+d', k_2+1]\\+ &\sum_{i = 1}^{d+1}\Phi(m_2)[k_1, i, k_2+1]\Phi(m_1)[i, k_1+d', k_2+1]
        \end{align*}
        By induction, we see that the first sum $\displaystyle\sum_{i = 1}^{d+1}\Phi(m_1)[k_1, i, k_2+1]\Phi(m_2)[i, k_1+d', k_2+1]$ is equal to 
        
        $$\left(\sum_{\pi\in \Sigma_{m_1}}\prod_{t = 1}^{d_1}z_{\pi(t), t + k_1 - 1, l^{m_1, \pi}_{t}+k_2}\right)\left(\sum_{\tau\in \Sigma_{m_2}}\prod_{t = 1}^{d_2}z_{\tau(t), t + d_1 + k_1 - 1, l^{m_2, \tau}_{t}+k_2}\right)$$
while the second sum $\sum_{i = 1}^{d+1}\Phi(m_2)[k_1, i, k_2+1]\Phi(m_1)[i, k_1+d', k_2+1]$ is equal to 

        $$\left(\sum_{\tau\in \Sigma_{m_2}}\prod_{t = 1}^{d_2}z_{\tau(t), t + k_1 - 1, l^{m_2, \tau}_{t}+k_2}\right)\left(\sum_{\pi\in \Sigma_{m_1}}\prod_{t = 1}^{d_1}z_{\pi(t), t + d_2 + k_1 - 1, l_{t}^{m_1, \pi}+k_2}\right)$$

        Expanding the products, we see that the first sum generates the monomials in the polynomial $\displaystyle\sum_{\sigma \in \Sigma_m} \prod_{t = 1}^{d'} z_{\sigma(t), t + k_1 - 1, l^{m, \sigma}_{t} + k_2 - 1}$ corresponding to the $\sigma\in\Sigma_{m}$ that make $m_1$ the left child of the root of $m$ and $m_2$ the right, and the second term generates the rest of the monomials. \\

        Also, note that the same argument as in the proof of Lemma \ref{lowerBound} tells us that if we have $i_1, i_2, i_3 $ such that $i_2 \neq i_1 + d'$ and $i_3 \in [d]$ then $\Phi(m)[i_1, i_2, i_3] = 0$ and that if $i_3>d-l+1$ then again $\Phi(m)[i_1, i_2, i_3]=0$.
    \end{proof}

    In particular, by setting $k_1, k_2=1$ in the statement of Claim \ref{polynomialEntry2} we see that for a monomial $m$ of degree $d'$, the $(1, d'+1, 1)$-th entry of $\Phi(m)$ is $\displaystyle\sum_{\sigma \in \Sigma_m} \prod_{t = 1}^{d'} z_{\sigma(t), t, l^{m, \sigma}_{t}}$. For a monomial of degree $\neq d'$, this entry is zero.
    
    Therefore, Combining Claim \ref{polynomialEntry2}, Lemma \ref{levelSeq} and Observation \ref{commutativeObs}, we see that $f$ cannot be an a PI for $\mathbb{C'}_d$ (and therefore for $\mathbb{C}_d$).
\end{proof}

The proof of Theorem \ref{Comm-NonAssociative-SZ} follows from Lemma \ref{lowerBound2}, in a way similar to the proof of Theorem \ref{NonAssociative-SZ}. We record the statement here for completeness

\commnonassocSZ*

\section{Deterministic $\pit$ Algorithms Over \texorpdfstring{$\cna$}{Lg}}\label{sec:det-algo}
In this section, we develop efficient deterministic $\pit$ algorithms over the algebra $\cna$. First, we present the white-box algorithm and then the black-box algorithm. 

\subsection{White-box Deterministic $\pit$ over \texorpdfstring{$\cna$}{Lg}}\label{WhiteboxPIT}

We give a polynomial time white-box $\pit$ algorithm for commutative, nonassociative circuits. We use linear algebraic ideas from the $\pit$ algorithm by Raz and Shpilka \cite{RS04} for noncommutative algebraic branching programs. These ideas have later been used to give polynomial time white-box $\pit$ algorithms for various models, such as nonassociative, \textit{noncommutative} circuits \cite{ADMR17} and noncommutative unique parse tree circuits \cite{Lagarde2019}.

\DetPIT*

\begin{proof}

For each monomial $m$ of degree $\leq d$, we may associate a vector $v_m \in \mathbb{F}^s$ where $s$ is the number of gates in $\Psi$. The vector $v_m$ is indexed by the gates of $\Psi$ such that $v_m(g) = {\sf coeff}_m(g)$. For each $i \in \{0\} \cup [d]$, we wish to maintain a polynomially bounded set $M_i$ of monomials of degree $i$ and a corresponding set $B_i = \{v_m\mid m \in M_i\}$ of vectors such that ${\sf span}\{B_i\} = {\sf span}\{v_m\mid {\sf deg}(m) = i\}$, and we build these sets inductively, starting from $i = 0, 1$. For $i = 0, 1$, we set $M_i$ to be the set of all monomials of degree $i$ and populate the vectors in $B_i$ in a brute force manner.

Next, suppose we have the sets $M_i, B_i$ for $0 \leq i < j$ and we want to construct $M_j$ and $B_j$ ($j \geq 2$). We set $$M'_j = \bigcup_{\substack{i + k = j \\ i, k \geq 1}}\{m\times m'\mid m \in M_i \text{ and }m' \in M_k\}$$ For all $m = m_1m_2\in M'_j$ we do the following: we sort the gates of $\Psi$ in topological order and fill the entries of $v_m$ in that order as follows:

\begin{itemize}

    \item If $g$ is a leaf, we set $v_{m}(g) = 0$ (since $j = {\sf deg}(m) \geq 2$).
    \item If $g = g_1 \times g_2$ is a product gate, set $v_{m}(g) = v_{m_1}(g_1)v_{m_2}(g_2) + v_{m_2}(g_1) v_{m_1}(g_2) + v_{m}(g_1)v_{1}(g_2) + v_{1}(g_1) v_m(g_2)$. We can do this since we know the vectors $v_{m_1}, v_{m_2}$ by induction on the degree, and we know $v_{m}(g_1), v_{m}(g_2)$ as $g_1, g_2$ appear before $g$ in the topological order.
    \item If $g = g_1 + g_2$ is a sum gate, set $v_m(g) = v_{m}(g_1) + v_{m}(g_2)$. Again, we know $v_{m}(g_1), v_{m}(g_2)$ as $g_1, g_2$ appear before $g$.
\end{itemize}


Finally, we select a maximal linearly independent subset $B_j$ (using Gaussian elimination) from the set of vectors $\{v_m\mid m \in M'_j\}$ and call the corresponding set of monomials $M_j$. Clearly, $|M_j| \leq s$.

\begin{claim}\label{span}
    For any monomial $m$ such that ${\sf deg}(m) = j$, $v_m \in {\sf span}\{B_j\}$.
\end{claim}

\begin{proof}
    The proof is by induction on the degree $j$. For $j = 0, 1$, the claim is trivially true, so assume $j \geq 2$. Now suppose $m$ is uniquely decomposed (up to commutativity) as $m = m_1m_2$ such that ${\sf deg}(m_1) = i$ and ${\sf deg}(m_2) = k$ (with $i, k\geq 1$). By induction on degree, we assume that 

    \begin{equation}\label{eq:1}
        v_{m_1} = \sum_{m' \in M_i}\alpha_{m'}v_{m'} \text{ and } v_{m_2} = \sum_{m'' \in M_k}\beta_{m''}v_{m''}
    \end{equation}

\noindent  for constants $\{\alpha_{m'}\mid m' \in M_i\}$ and $\{\beta_{m''}\mid m'' \in M_k\}$. We will prove that for each gate $g$ of $\Psi$

    \begin{equation} \label{eq:2}
        v_{m}(g) = \sum_{\substack{m' \in M_i \\ m'' \in M_k}}\alpha_{m'}\beta_{m''}v_{m'm''}(g)
    \end{equation}

     This puts $v_m$ in ${\sf span}\{B_j\}$ by construction. We prove (\ref{eq:2}) gate by gate, by induction on the depth of the gate.

     \begin{itemize}
         \item For a leaf $g$, $g$ is labeled by a variable or a constant, so $v_{m}(g) = 0$ (recall that $j\geq 2$) and $v_{m'm''}(g)$ for each $m' \in B_i, m'' \in B_k$ is also all zero by construction, so (\ref{eq:2}) is true in this case.
         \item For a product gate $g = g_1 \times g_2$, \begin{align*}
             v_{m}(g) &= v_{m_1}(g_1)v_{m_2}(g_2) + v_{m_2}(g_1)v_{m_1}(g_2) + v_{m}(g_1)v_1(g_2) + v_{1}(g_1)v_{m}(g_2).
         \end{align*}
        After substituting (\ref{eq:1}) for $v_{m_1}$ and $v_{m_2}$ using the (degree) induction hypothesis, substituting (\ref{eq:2}) for $v_{m}(g_1), v_{m}(g_2)$ using induction on the depth of $g_1, g_2$ and simplifying, we get that 
        
             $$v_{m}(g) = \displaystyle\sum_{\substack{m' \in M_i \\ m'' \in M_k}}\alpha_{m'}\beta_{m''} \left(\sum_{(r,t) \in \{(1, 2), (2, 1)\}}v_{m'}(g_r)v_{m''}(g_t)+v_{m'm''}(g_r)v_{1}(g_t)\right)$$ 
        Note that the inner expression is nothing but $v_{m'm''}(g)$ (by construction), and therefore (\ref{eq:2}) is true when $g$ is a product gate.

        \item For a sum gate $g = g_1 + g_2$, we have
        \begin{align*}
            v_m(g) &= v_{m}(g_1) + v_{m}(g_2) \\ &= \sum_{\substack{m' \in M_i \\ m'' \in M_k}}\alpha_{m'}\beta_{m''}v_{m'm''}(g_1) + \sum_{\substack{m' \in M_i \\ m'' \in M_k}}\alpha_{m'}\beta_{m''}v_{m'm''}(g_2) \\ &= \sum_{\substack{m' \in M_i \\ m'' \in M_k}}\alpha_{m'}\beta_{m''}(v_{m'm''}(g_1) + v_{m'm''}(g_2))\\ &=  \sum_{\substack{m' \in M_i \\ m'' \in M_k}}\alpha_{m'}\beta_{m''}v_{m'm''}(g)
        \end{align*}
        where the second step follows by induction on depth and the fourth by construction of $v_{m'm''}$. This verifies (\ref{eq:2}) when $g$ is a sum gate.
    \end{itemize}
    These three cases together prove the claim.
\end{proof}

Clearly, this procedure of constructing $B_j$'s and $M_j$'s takes polynomial time since $|B_j|\leq {\max\{s, n, 1\}}$ for each $0 \leq j \leq d$. To check whether $\Psi \equiv 0$, we simply check whether there exists an $i \in \{0, \ldots, d\}$ such that there exists a monomial $m\in M_i$ such that $v_m(g_s) \neq 0$, where $g_s$ is the output gate of $\Psi$. If yes, we say $\Psi \not\equiv 0$, otherwise we say $\Psi \equiv 0$.
\end{proof}


\subsection{Deterministic Nonassociative Black-box $\pit$}\label{deterministicBlackboxPIT}

In this section, we design black-box $\pit$ algorithms for polynomials in $\cna$ and $\ncna$ computed by low depth circuits. Our algorithms run in quasipolynomial time as long as the input circuit $\Psi$ has depth \textit{polylogarithmic} in the size of $\Psi$ and the degree of the polynomial computed by it. The algorithms query $\Psi$ on elements of the algebra $\mathbb{A}_d$ in the non-commutative case and on elements of $\mathbb{C}_d$ in the commutative case, defined in sections \ref{nonassnoncomm} and \ref{blackbox-commutative} respectively. In particular, we construct \textit{hitting sets} of quasipolynomial size in both the commutative and noncommutative setting, for circuits that have depth polylogarithmic in its size and degree.

\hittingSet*

Over the algebra $\ncna$, we have the following analogous result.

\begin{theorem}\label{noncommutative-hittingset}
    There exists a set $H_{n, s, d, \Delta}\subseteq (\mathbb{A}_d)^n$ of size $(nsd)^{O(\Delta)}$ of points in $(\mathbb{A}_d)^n$ such that for every nonassociative, noncommutative circuit $\Psi$ of size $\leq s$ and product depth $\leq \Delta$ computing a non-zero polynomial $f\in\ncna$ of degree $\leq d$, there is a point in $H_{n, s, d, \Delta}$ at which $f$ is non-zero. Furthermore, we can compute $H_{n, s, d, \Delta}$ in time $(nsd)^{O(\Delta)}$.
\end{theorem}

We prove Theorem \ref{commutative-hittingset} and Theorem \ref{noncommutative-hittingset} in the next two subsections. The idea is to reduce PIT for nonassociative circuits to $\pit$ for associative, unambiguous circuits (see below for a formal definition) and then give hitting sets for unambiguous circuits.

\begin{defn}[Unambiguous Circuits]
Let $Z = \{z_1. \ldots, z_n\}$ be a commutative, associative set of variables and let $\text{mons}(Z)$ denote the set of monomials in the variables in $Z$. We say that a circuit $\Psi$ computing a polynomial $f\in\mathbb{F}[Z]$ is $\textit{unambiguous}$ if for each monomial $m$ in $f$ there is a {\em unique} reduced parse tree $T_m$ in $\Psi$ generating $m$. That is, if another reduced parse tree $T$ generates $m$ at any gate in $\Psi$, then the trees $T$ and $T_m$ are identical as labeled rooted binary trees. 
\end{defn}


\subsubsection{Reduction from non-associative to associative, unambiguous circuits}\label{reduction}
The first, fairly straightforward step is to reduce PIT for nonassociative circuits to PIT for associative, unambiguous circuits. We do this via a set-multilinearization argument. This type of reduction is commonplace in noncommutative PIT literature (see for example \cite{FS13}, \cite{ST18})

We first describe the reduction in the nonassociative, commutative setting. The proof is along exactly the same lines in the noncommutative setting as well (in fact it is simpler).

Let $\{x_1, \ldots, x_n\}$ be a set of variables. Let $\Psi$ be a circuit computing a polynomial $f\in \cna$ of degree $d' \leq d$. For each variable $x_i, i\in[n]$, we introduce a set $\{z_{i,j,k}\mid 1\leq j,k \leq d\}$. Let $Z=\bigcup_{i=1}^{n}\{z_{i,j,k}\mid 1\leq j,k \leq d\}$ and define the algebras $\mathbb{C}_d$ and $\mathbb{C'}_d$ over the field $\mathbb{F}(Z)$. We consider the evaluation map $\Phi:\cna\rightarrow \mathbb{C'}_d$ defined in Section \ref{blackbox-commutative}. Let us recall the definition of $\Phi$: $\Phi$ maps $x_i$ to $Z_i$ where for each $1\leq j,k\leq d$, the  entry $(j, j+1, k)$ of $Z_i$ is $z_{i, j, k}$ and the rest of the entries are zero. 

We will examine the image of the circuit $\Psi$ under $\Phi$. To this end, we define another map $\phi: \cna\rightarrow\mathbb{F}[Z]$. Let $m$ be a monomial in $\{x_1, \ldots, x_n\}$ of degree $d'$. Recall that in Section \ref{blackbox-commutative}, we associated to $m$ a set $\Sigma_m$ of ``orders" and for each $\sigma\in\Sigma_m$, we had $l^{m, \sigma}_{t}$ (for each $1\leq t \leq d'$), the level at which the $t$-th variable in the order $\sigma$ appears in $m$. For a monomial $m$ in $\cna$ define $$\phi(m)\triangleq\displaystyle\sum_{\sigma \in \Sigma_m} \prod_{t = 1}^{d'} z_{\sigma(t), t, l^{m, \sigma}_{t}}$$
and extend $\phi$ linearly to all of $\cna$.

\begin{lemma}\label{set-mult}
    Let $f\in \cna$ be a homogeneous polynomial of degree $d'\leq d$. The $(1, d'+1, 1)$-th entry of $\Phi(f)$ is equal to $\phi(f_{d'})$ where $f_{d'}$ is the homogeneous degree $d'$ component of $f$.
\end{lemma}

\begin{proof}
    This is a simple corollary of Claim \ref{polynomialEntry2}. Setting $k_1, k_2 = 1$ in the statement of Claim \ref{polynomialEntry2}, we see that for a monomial $m$ of degree $d'$, the $(1, d'+1, 1)$-th entry of $\Phi(m)$ is $\phi(m)$. On the other hand, for a monomial $m$ of degree $\neq d'$, we get that the $(1, d'+1, 1)$-th entry of $\Phi(m)$ is $0$. Therefore, Lemma \ref{set-mult} follows by summing over monomials of $f$. 
\end{proof}

\begin{lemma}\label{unambiguous-construction-commutative}
    Let $\Psi$ be a nonassociative, commutative arithmetic circuit of size $s$ and product depth $\Delta$ computing a polynomial $f\in \cna$ of degree $d'\leq d$. Then, there exists an unambiguous circuit $\Psi'$ computing $\phi(f_{d'})$ where $f_{d'}$ is the homogeneous degree $d'$ component of $f$. Furthermore, the size of $\Psi'$ is at most $3d^4s$ and product depth at most $\Delta$.
\end{lemma}

\begin{proof}

We assume without loss of generality that we have a homogeneous circuit $\Psi^{d'}$ for computing $f_{d'}$, with size at most $d^2s$ and product depth at most $\Delta$. 
We can do this since nonassociative circuits can be homogenized using a standard technique \cite{SY10} with a multiplicative blowup of $d^2$ in size. This has also been observed by the authors in \cite{HWY10}. Furthermore, we may also assume (without loss of generality) that $\Psi^{d'}$ does not have gates (including leaves) computing constants from $\mathbb{F}$.
    
We will build $\Psi'$ using $\Psi^{d'}$, in a bottom up fashion such that $\Psi'$ computes the non-zero entries of $\Phi(f_{d'})$. First, for each leaf labeled by $x_i$ in $\Psi^{d'}$, introduce $d^2$ leaves in $\Psi'$, each labeled by one of the $d^2$ nonzero entries ($\{z_{i, j, k}\mid j, k\in[d]\}$) of $Z_i$. Recall that $z_{i, j, k}$ appears as the $(j, j+1, k)$-th entry of $\Phi(x_i)=Z_i$. Next, suppose we are at an internal gate $g$ of $\Psi^{d'}$ with children $g_1, g_2$. Let $f^{g}$ denote the polynomial computed at gate $g$ in $\Psi^{d'}$ and let the degree of $f^{g}$ be $d''\leq d'$.
\begin{itemize}
    \item \textit{$g$ is a sum gate in $\Psi^{d'}$}. In this case, ${\sf deg}(f^{g_1})={\sf deg}(f^{g_2})=d''$ since $\Psi^{d''}$ is homogeneous. For each $k_1\in [d-d''+1]$ and $k_2\in [d]$, we compute in $\Psi'$ the $(k_1, k_1+d'', k_2)$-th entry of $\Phi(f^g)$ by summing the corresponding entries of $\Phi(f^{g_1})$ and $\Phi(f^{g_2})$, since sum in $\mathbb{C'_d}$ is pointwise. Note that these are the only entries of $\Phi(f^g)$ that can be nonzero, by Claim \ref{polynomialEntry2}.
    \item \textit{$g$ is a product gate in $\Psi^{d'}$ with children $g_1, g_2$:} Let the degrees of $f^{g_1}, f^{g_2}$ be $d_1, d_2$ respectively with $d_1+d_2=d''$. For each $k_1\in [d-d''+1]$ and $k_2\in [d]$, we compute in $\Psi'$ the $(k_1, k_1+d'', k_2)$-th entry of $\Phi(f^g)$ as follows:    \begin{align*}
        \Phi(f^{g})[k_1, k_1+d'',k_2]=&\Phi(f^{g_1})[k_1, k_1+d_1, k_2+1]\Phi(f^{g_2})[k_1+d_1, k_1+d'', k_2+1]+\\ &\Phi(f^{g_2})[k_1, k_1+d_2, k_2+1]\Phi(f^{g_1})[k_1+d_2, k_1+d'', k_2+1]
    \end{align*}
    Again, this is justified by Claim \ref{polynomialEntry2}. The entries we have computed are the only entries of $\Phi(f^g)$ that can be nonzero.
\end{itemize}

By Lemma \ref{set-mult}, the gate in $\Psi'$ computing the $(1, d'+1, 1)$-th entry of $\Phi(f_{d'})$ computes $\phi(f_{d'})$. By construction, $\Psi'$ has size $3d^4 s$ and product depth at most $\Delta$. 

\begin{claim}
    $\Psi'$ is unambiguous.
\end{claim}
\begin{proof}
    This follows from the fact that $\Psi^{d'}$ is nonassociative. In particular, owing to Claim \ref{polynomialEntry2}, we have that for every monomial $m\in\mathbb{F}[Z]$ (with degree $d''\leq d$') computed by $\Psi'$, there exists a unique monomial $\hat{m}\in\cna$ with degree $d''$ and depth $l$, an order $\sigma \in \Sigma_{\hat{m}}$, a $k_1\in[d-d''+1]$ and a $k_2\in[d-l+1]$ such that $$m=\prod_{t=1}^{d'}z_{\sigma(t), t+k_1-1,l_{t}^{\sigma, \hat{m}}+k_2-1}$$ Furthermore, if the reduced parse tree for $m$ decomposes $m$ as $m=m_1m_2$ then $\hat{m}=\hat{m_1}\hat{m_2}$. This property ensures that $\Psi'$ is indeed unambiguous
\end{proof}
This finishes the proof of Lemma \ref{unambiguous-construction-commutative}.
\end{proof}

In the non-commutative setting, let $\{x_1, \ldots, x_n\}$ be a set of variables and let the set $Z$ of variables be as before. We consider the map $\Phi:\ncna\rightarrow \mathbb{A'}_d$ from Section \ref{nonassnoncomm} (where $\mathbb{A'}_d$ is an algebra over the field $\mathbb{F}(Z)$). In this case, $\Phi$ sends $x_i$ to $Z_i\in \mathbb{A'}_d$ where for each $1\leq j, k\leq d$ the $(j, j+1, k)$-th entry of $Z_i$ is $z_{i, j, k}$ and every other entry is zero. 

Recall that for any monomial $m\in \ncna$ of degree $d'$ there is a unique left to right order $\sigma_m:[d']\rightarrow [n]$ of variables associated to $m$. To study the image of a polynomial under $\Phi$, we need the auxiliary map $\phi:\ncna\rightarrow\mathbb{F}[Z]$ defined as follows: for a monomial $m$ as above, 

$$\phi(m)=\prod_{t=1}^{d'}z_{\sigma_{m}(t), t, l^m_t}$$

Using these definitions of $\Phi$ and $\phi$, and along the same lines as the proofs of Lemmas \ref{set-mult} and \ref{unambiguous-construction-commutative}, we obtain the following:

\begin{lemma}\label{set-mult-noncomm}
    Let $f\in \ncna$ be a homogeneous polynomial of degree $d'\leq d$. The $(1, d'+1, 1)$-th entry of $\Phi(f)$ is equal to $\phi(f_{d'})$ where $f_{d'}$ is the homogeneous degree $d'$ component of $f$.
\end{lemma}

\begin{lemma}\label{unambiguous-construction-noncommutative}
    Let $\Psi$ be a nonassociative, noncommutative arithmetic circuit computing a polynomial $f\in \ncna$ of degree $d'\leq d$. Then, there exists an unambiguous circuit $\Psi'$ computing $\phi(f_{d'})$ (where $f_{d'}$ is the homogeneous degree $d'$ component of $f$). Furthermore, the size of $\Psi'$ is at most $3d^4s$ and product depth at most $\Delta$.
\end{lemma}

\subsubsection{Hitting sets for low-depth associative, unambiguous circuits}

In this section, we construct hitting sets for {\em low-depth} unambiguous circuits in the commutative, associative setting. We first define the most important tool in the design of hitting sets for unambiguous circuits: {\em basis isolating weight assignments}. Agrawal et al.\ \cite{AGKS15} defined basis isolating weight assignments and used them to construct hitting sets for Read-Once oblivious ABPs. These weight assignements were subsequently also used in a work of Saptharishi and Tengse \cite{ST18} for PIT of noncommutative unique parse tree circuits.

Let $Z=\{z_1, \ldots, z_n\}$ and let $\Psi$ an unambiguous circuit with $s$ gates computing a polynomial $f\in\mathbb{F}[Z]$ of degree $d$. Define $f_\Psi\in\mathbb{F}^{s}[Z]$ to be the polynomial $\sum\limits_{m\in \text{mons}(Z)}v_m m$ where for each monomial $m$, $v_m$, the coefficient of $m$ in $f_\Psi$, is (as before) an $s$ dimensional vector whose entries are indexed by gates of $\Psi$ and for each gate $g$ in $\Psi$, $v_{m}(g) \triangleq \text{coeff}_m(g)$. Let $w:Z\rightarrow \mathbb{N}$ be a weight function that assigns weights to variables in $Z$. $w$ extends to ${\text{mons}(Z)}$ naturally as follows: $w(z_1^{i_1}z_2^{i_2}\ldots z_n^{i_n}) = \sum_{j = 1}^{n}i_jw(z_j)$. 

\begin{defn}[Basis Isolating Weight Assignment]
   A weight function $w: Z\rightarrow \mathbb{N}$ is said to be a \textit{basis isolating weight assignment} for $f_\Psi\in\mathbb{F}^{s}[Z]$ if there exists a set $M$ of {\em isolated} monomials in $\text{mons}(Z)$ such that the following conditions hold:
\begin{enumerate}
    \item $B = \{v_m\mid m\in M\}$ forms a basis for $\text{span}\{v_m\mid m\in\text{mons}(Z)\}$
    \item For $m, m' \in M$ such that $m \neq m'$, we have that $w(m)\neq w(m')$
    \item For each $m\not\in M$, $v_m\in\text{span}\{v_{m'}\mid m'\in M, w(m')<w(m)\}$.
\end{enumerate} 
\end{defn}

In order to build basis isolating weight assignments, we need an efficient version of the Kronecker map described in \cite{AB03}, \cite{AGKS15} that we now state.

\begin{lemma}[Efficient Kronecker map, \cite{AB03}]
\label{lemma:splitting}
Let $Z = \{z_1, \ldots, z_n\}$ be a set of commutative, associative variables. For each $k,d \geq 1$, there is a set $W_{k,d}$ of $N \leq n\binom{k}{2}\log(d + 1)$ weight functions $w:Z\rightarrow [2N\log N]$ such that for any set $A$ of monomials in $Z$ of individual degree $\leq d$ satisfying $|A| \leq k$, there exists a $w\in W_{k,d}$ that \textbf{separates} $A$, that is, $\forall m\neq m'\in A$, $w(m)\neq w(m')$. Furthermore, the set $W_{k,d}$ is constructible in polynomial time.
\end{lemma}

For convenience, we say that a set $W$ of weight assignments to $Z$ separates a set $A$ of monomials if there exists a $w\in W$ that separates $A$. In what follows, we will construct a basis isolating weight assignment $w$ for $f_\Psi$. 

\begin{theorem}\label{BIWA}
    Let $\Psi$ be an unambiguous circuit with $s$ gates
    and product depth $\Delta$ computing $f\in\mathbb{F}[Z]$ of degree $d$. Let $f_{\Psi} = \sum_{m\in\text{mons}(Z)}v_mm$ be as above. Then, we can construct a basis isolating weight assignment $w:Z\rightarrow\mathbb{N}$ for $f_{\Psi}$ such that $w(z_i) = (nds)^{O(\Delta)}$, for each $i\in [n]$. Furthermore, we can construct $w$ in time $(nds)^{O(\Delta)}$.
\end{theorem}

\begin{proof}
For any monomial $m$, the {\em depth} of $m$ (with degree $\geq 1$) in $\Psi$ is the depth of $T_m$, the unique reduced parse tree in $\Psi$ computing the monomial $m$. Since the product depth of $\Psi$ is at most $\Delta$, we have that for each monomial $m$, the depth of $T_m$ is also $\leq\Delta$. 
Let $f_{\Psi} = \sum\limits_{m\in\text{mons}(Z)}v_mm$ where $v_m\in \mathbb{F}^s$ is as defined before. The basis isolating weight assignment $w$ for $f_{\Psi}$ is obtained by taking a carefully chosen linear combination of $\Delta+1$ many weight functions. That is, the $\Delta+1$ many weight functions are chosen and scaled appropriately to obtain the final basis isolating weight function as in Agrawal et al.\ \cite{AGKS15}. Recall that $W_{k,d}$ is the set of weight assignments from Lemma \ref{lemma:splitting} that separates every set of at most $k$ monomials of degree $\leq d$.

\begin{claim}
\label{claim:existence}
There exist $w_2, \ldots, w_{\Delta+1}\in W_{s^2\Delta, d}$ such that $$w \triangleq \sum_{i = 0}^{\Delta} B^{\Delta - i}w_{i+1}$$
is a basis isolating weight assignment for $f_\Psi$, where $w_1$ sends each $z_i$ to the value $i$, and $B = 1 + \max\{w_i(m)\}$ where the $\max$ is over all monomials $m$ of degree $\leq d$ in $Z$, and all $i$ belonging to $[\Delta+1]$. 
\end{claim}
\begin{proof}
We construct the weight function $w$ by iteratively constructing $w_1, \ldots ,w_{\Delta}$. For this, we will need the following function:
\begin{equation}
w^j \triangleq \sum_{i = 0}^{j-1} B^{j-1 - i}w_{i+1}
\end{equation}
First, we sort the variables in $Z$ in ascending order of their weight with respect to $w_1$ and pick  a basis $B_1$ for ${\sf span}\{v_{z_i}\mid i \in [n]\}$ greedily starting from the lowest weight variable. Let the corresponding set of monomials (variables) be $M_1$. Clearly, $|M_1|\leq s$ and $w^1 = w_1$. Further, by construction of $B_1$ and $M_1$ and as $w^1 = w_1$, we have that for every monomial $m$ of depth $1$, $v_m\in{\sf span}\{v_{m'}\mid m'\in M_1, w^1(m')<w^1(m)\}$.  Now, we extend this idea to monomials of depth $j \in \{2, \ldots, \Delta \}$. More precisely, we show the following:\\

\noindent For each $2 \leq i \leq \Delta$, there exists $w_i\in W_{s^2\Delta, d}$ such that for each $j = 1, \ldots, \Delta$ there exists a subset $M_j$ of monomials of depth $j$ and a corresponding set $B_j = \{v_m \mid m\in M_j\}$ of coefficients of $f_{\Psi}$ such that for every monomial $m$ of depth $j$, $v_m\in{\sf span}\{v_{m'}\mid m'\in M_j, w^j(m')<w^j(m)\}$. Furthermore,  $|M_j| \leq s$ for each $j\in [\Delta]$, and $w^j$ separates $M_j$. \\

\noindent We will prove this by induction on $j$. Observe that the base case $j=1$ has already been discussed above. Suppose we  already have weight functions $w_1, \ldots, w_{j - 1}$ (with $w_2, \ldots, w_{j-1}\in W_{s^2\Delta,d}$) satisfying our requirements. Then we know that there exist corresponding sets $M_1, \ldots, M_{j-1}$ and $B_1, \ldots, B_{j-1}$ also satisfying the above mentioned conditions. Define $M'_j$ and $B'_j$ as follows:
\begin{align*}
M'_{j}&\triangleq\bigcup_{k \leq j-1}\left\{
  m'\cdot m''\;\middle|\;
  \begin{aligned}
  & m'\in M_{j - 1},m''\in M_k, {\sf depth}(T_{m'm''})=j,\\
  & T_{m'm''} \text{ decomposes }m'\cdot m''\text{ as }m'\times m''
  \end{aligned}
\right\}\\
B'_j &\triangleq \{v_{m}\mid m \in M'_j\}
\end{align*}
where we say that a parse tree $T$ for a monomial $m=m_1\cdot m_2$ {\em decomposes} $m$ as $m_1\times m_2$ if one of the children of the root of $T$ computes $m_1$ and the other computes $m_2$. 

Let $m$ be any monomial of depth $j$ with reduced parse tree $T_m$ such that $T_m$ decomposes $m$ as $m=m_1\times m_2$ with ${\sf depth}(m_1) = j-1$ and ${\sf depth}(m_2)=k\leq j-1$. Let us define $M_{m_1}$ and $M_{m_2}$ as follows: 
\begin{align*}
M_{m_1} &\triangleq \{ m'\in M_{j-1} \mid   w^{j-1}(m') < w^{j-1}(m_1) \} \subseteq M_{j-1} \\
M_{m_2} &\triangleq \{ m'' \in M_k \mid   w^{k}(m'') < w^{k}(m_2) \}  \subseteq M_k
\end{align*}
Let us also define the set $M_{m_1m_2}\subseteq M_j'$ as follows:

\begin{equation*}
 M_{m_1m_2}\triangleq\left\{
  m'\cdot m''\;\middle|\;
  \begin{aligned}
  & m'\in M_{m_1},m''\in M_{m_2}, {\sf depth}(T_{m'm''})=j,\\
  & T_{m'm''} \text{ decomposes }m'\cdot m''\text{ as }m'\times m''
  \end{aligned}
\right\}
\end{equation*}

We will show that for each monomial $m$ as defined above, we have 
\begin{equation}
\label{eq:vspan}
    v_{m}\in{\sf span}\{v_{m'm''}\mid m'\cdot m''\in M_{m_1m_2}\}
\end{equation}

The proof of Equation \ref{eq:vspan} is similar to the proof of Claim \ref{span}, although in this case we will need to be careful about the depth of the monomials involved. By the induction hypothesis, we have that $v_{m_1}\in{\sf span}\{v_{m'}\mid  m'\in M_{m_1}\}$ and  $v_{m_2}\in{\sf span}\{v_{m''}\mid  m''\in M_{m_2}\}$. That is, 

\begin{eqnarray}
     v_{m_1}&= \sum\limits_{m'\in M_{m_1}}\alpha_{m'}v_{m'} \label{eq:m1} \\
     v_{m_2}&=\sum\limits_{m''\in M_{m_2}}\beta_{m''}v_{m''} \label{eq:m2}
\end{eqnarray}

where the $\alpha_{m'}$'s and $\beta_{m''}$'s are scalars in $\mathbb{F}$. In order to prove Equation \ref{eq:vspan}, we will show that for each gate $g$ in $\Psi$, 

\begin{equation}\label{eq:depthSpan}
    v_{m_1m_2}(g) = \sum_{\substack{m'\in M_{m_1}, m''\in M_{m_2}\\m'm''\in M_{m_1m_2}}}\alpha_{m'}\beta_{m''}v_{m'm''}(g)
\end{equation}

We do this, as in the proof of Claim \ref{span}, by induction on the depth of the gate $g$ in $\Psi$.

\begin{enumerate}
    \item If $g$ is a leaf, then both the LHS and RHS in Equation \ref{eq:depthSpan} are $0$ (since $j\geq 2$). Therefore, Equation \ref{eq:depthSpan} is true in this case.
    \item If $g=g_1+g_2$ is a sum gate, we have $$v_{m_1m_2}(g)=v_{m_1m_2}(g_1)+v_{m_1m_2}(g_2)$$ By induction on the depth of the gates $g_1, g_2$, Equation \ref{eq:depthSpan} is true for $v_{m_1m_2}(g_1)$ and $v_{m_1m_2}(g_2)$ and therefore it is also true for $v_{m_1m_2}(g)$.
    \item The interesting case is when $g=g_1\times g_2$ is a product gate. In this case, $$v_{m_1m_2}(g)=v_{m_1}(g_1)v_{m_2}(g_2)+v_{m_1}(g_2)v_{m_2}(g_1)+v_{m_1m_2}(g_1)v_1(g_2)+v_{1}(g_1)v_{m_1m_2}(g_2)$$
    Substituting (\ref{eq:m1}) for $v_{m_1}$ and (\ref{eq:m2}) for $v_{m_2}$ we get

    \begin{equation}\label{eq:sum}
    v_{m_1}(g_1)v_{m_2}(g_2)+v_{m_1}(g_2)v_{m_2}(g_1)=\sum_{\substack{m'\in M_{m_1}\\m''\in M_{m_2}}}\alpha_{m'}\beta_{m''}(v_{m'}(g_1)v_{m''}(g_2)+v_{m'}(g_2)v_{m''}(g_1))
    \end{equation}

    The key observation here is that if for some $m'\in M_{m_1}$ and $m''\in M_{m_2}$ it is the case that $m'm''\not\in M_{m_1m_2}$, then $v_{m'}(g_1)v_{m''}(g_2)=0$ and $v_{m'}(g_2)v_{m''}(g_1)= 0$, for otherwise $T_{m'm''}$ it would decompose $m'm''$ as $m'\times m''$ and it would have depth $j$ (because ${\sf depth}(T_{m'})=j-1$ and ${\sf depth}(T_{m''})=k\leq j-1$).

    Therefore, we have that the only surviving terms in the RHS of Equation \ref{eq:sum} are those corresponding to the $m'\in M_{m_1}$ and $m''\in M_{m_2}$ such that $m'm''\in M_{m_1m_2}$. Also, by induction on the depth of the gates $g_1$ and $g_2$, Equation \ref{eq:depthSpan} is true for $g_1$ and $g_2$. Combing these observations, we see that $v_{m_1m_2}(g)$ is equal to $$\sum_{\substack{m'\in M_{m_1}, m''\in M_{m_2}\\m'm''\in M_{m_1m_2}}}\alpha_{m'}\beta_{m''}(v_{m'}(g_1)v_{m''}(g_2)+v_{m'}(g_2)v_{m''}(g_1)+v_{m}(g_1)v_{1}(g_2)+v_{1}(g_1)v_{m}(g_2))$$ This quantity is exactly equal the RHS of Equation \ref{eq:depthSpan}.
\end{enumerate}
These three cases together prove Equation \ref{eq:depthSpan}, and therefore, Equation \ref{eq:vspan}. Note that Equation \ref{eq:vspan} puts $v_m$ in the span of $B'_j$ by definition of $M_j'$. \\

\noindent Furthermore, observe the following:
\begin{itemize}
\item For any $m'' \in M_{m_2}$, $w^{k}(m'')< w^{k}(m_2)$, by the definition of set $M_{m_2}$.
\item  For any two monomials $\tilde{m_1}, \tilde{m_2}$ and any $1 \leq l\leq t\leq \Delta$, 
 $w^l(\tilde{m_1})<w^l(\tilde{m_2})\implies w^{t}(\tilde{m_1})<w^{t}(\tilde{m_2})$.  This is by the choice of $B$: $w^{t}(\tilde{m_1})=B^{t-l}w^{l}(\tilde{m_1})+B^{t-l-1}w_{l+1}(\tilde{m_1})+\ldots+w_{t}(\tilde{m_1})\leq B^{t-l}w^l(\tilde{m_1})+B^{t-l}-1 < B^{t-l}w^{l}(\tilde{m_2})\leq w^t(\tilde{m_2})$.
\item Therefore we have that for any $m''\in M_{m_2}$, $w^{j-1}(m'')< w^{j-1}(m_2)$ (since $k\leq j-1$).
\end{itemize}

From the observations above, we see that for any $m'\in M_{m_1}, m''\in M_{m_2}$ we have:
\begin{equation}
\label{eq:weight}
w^{j-1}(m)= w^{j}(m_1\cdot m_2) = w^{j-1}(m_1)+w^{j-1}(m_2)>w^{j-1}(m')+w^{j-1}(m'')=w^{j-1}(m'\cdot m'')
\end{equation}
 Hence, from Equations (\ref{eq:vspan}) and (\ref{eq:weight}), for any monomial $m$ of depth $j$ we have that $v_m\in{\sf span}\{v_{m'}\mid m'\in M'_{j}, w^{j-1}(m')<w^{j-1}(m)\}$. To complete the induction step, pick a $w_j\in W_{s^2\Delta,d}$ that separates $M_j'$ (this fixes $w^j$ as well). By the way $w^j$ is defined, $w^j$ also separates $M_j'$. Sort the elements of $M_j'$ (and therefore $B'_j$) in ascending order of weight with respect to $w^j$ and pick a greedy basis for ${\sf span}\{B_j'\}$, going over vectors in $B_j'$ from left to right starting with the lowest weight monomial. Set that basis to be $B_j$ and the corresponding set of monomials to be $M_j$. Clearly, for all monomials $m$ of depth $j$, $v_m\in{\sf span}\{v_{m'}\mid m'\in M_j, w^j(m')<w^j(m)\}$, and $|M_j| \leq s$. Also, $w^j$ separates $M_j$ by construction. This completes the inductive step.

        At $j=\Delta$, we will have constructed $w^{\Delta}$ with the following properties: 

        \begin{itemize}
            \item[(a)] $w^\Delta$ separates $M_j$ for all $1 \leq j \leq \Delta$. This is because $w^{j}$ separates $M_j$ and therefore so do all $w^k$ for $j \leq k \leq \Delta$.
            \item[(b)]\label{obs(b)} For all $j\in [\Delta]$ and every monomial $m$ of depth $j$, $v_m\in {\sf span}\{v_{m'}\mid w^{\Delta}(m')<w^{\Delta}(m)\}$. Again, this is because $w^j(m')<w^{j}(m)$ implies $w^{k}(m')<w^{k}(m)$ for all $j \leq k \leq \Delta$.     
\end{itemize}
 Consider the set $M'=\bigcup\limits_{j=1}^{\Delta} M_j$ (note that $|M'|\leq s\Delta$). Pick a $w_{\Delta+1}\in W_{s^2\Delta,d}$ that separates $M'$. Consider the weight assignment $w = w^{\Delta+1}\triangleq Bw^{\Delta}+w_{\Delta+1}$. It is not hard to see that $w$ is a basis isolating weight assignment for $f_{\Psi}$. Sort the monomials in $M'$ in ascending order of weight with respect to $w^{\Delta+1}$, and pick a greedy basis for ${\sf span}\{v_m\mid m\in M'\}$, as before. Call this set $B$ and the corresponding set of monomials $M$. $M$ is the set of monomials \textit{isolated} by $w^{\Delta+1}$:
\begin{enumerate}
\item $w^{\Delta+1}$ separates $M$.
\item Let $m$ be any monomial. Let depth of $m$ be $j$. Then by observation (b) above, know that $v_m\in \{v_{m'}\mid m'\in M_j, w^{\Delta}(m')<w^{\Delta}(m)\}$. This implies that $v_m\in \{v_{m'}\mid m'\in M_j, w^{\Delta+1}(m')<w^{\Delta+1}(m)\}$.  But for each $m'\in M_j\subseteq M'$, $v_{m'}\in {\sf span}\{v_{m''}\mid m''\in M, w^{\Delta+1}(m'')<w^{\Delta+1}(m')\}$. Therefore, for each monomial $m$, we have $v_{m}\in {\sf span}\{v_{m''}\mid m''\in M, w^{\Delta+1}(m'')<w^{\Delta+1}(m)\}$. 

\item Since the vectors in $B$ are linearly independent, $B$ is in fact a basis for $\{v_m|m\in \text{mons}(Z)\}$.
\end{enumerate}   
This finishes the proof of Claim \ref{claim:existence}. \end{proof}
Now, having proved existence, in order to construct the basis isolating weight assignment from Claim \ref{claim:existence}, simply try all tuples $(w_2, \ldots, w_{\Delta+1})$ in $(W_{s^2\Delta,d})^\Delta$. At least one of them is sure to work. The cost of doing this is $\poly(n, s, d)^{\Delta}$, and the weight assignments arising from these tuples give at most $\poly(n,s,d)^\Delta$ weight to any variable.
\end{proof}

Next, we will show that a basis isolating weight assignment is \textit{useful} for PIT. 

\begin{lemma}\label{BIWA-map}
    Let $\Psi$ be a circuit computing $f(z_1, \ldots, z_n)\in \mathbb{F}[Z]$ and $f_\Psi = \displaystyle\sum_{m \in \text{mons}(Z)}v_mm$ be as defined earlier. Suppose $w$ is a basis isolating weight assignment for $f_\Psi$. Let $\phi$ be the polynomial map that sends $z_i$ to $t^{w(z_i)}$ where $t$ is a new variable. Then $\phi(f) \not\equiv 0 \iff f \not\equiv 0$. 
\end{lemma}

\begin{proof}
Clearly, $f\equiv 0 \implies \phi(f)\equiv 0$. For the other direction, first notice that $\phi(m) = t^{w(m)}$ for any monomial $m$. Let $M\subseteq\text{mons}(Z)$ be the set of monomials isolated by $w$. Let $B=\{v_m\mid m\in M\}$. Let $g$ be the output gate of $\Psi$ which computes $f$. If $f \not\equiv 0$ then $f$ must contain a monomial from $M$ with non-zero coefficient, because $B$ is a basis for ${\sf span}\{v_m\mid m\in \text{mons}(Z)\}$. Let $M'\subseteq M$ be the monomials from $M$ that occur in $f$ with non-zero coefficient. Let $m \triangleq {\sf argmin}_{m\in M'}\{w(m)\}$. $m$ is the unique minimizer, since $w$ separates $M$. Now suppose for $m'\not\in M$, we have that $w(m') = w(m)$. We know that $v_{m'}\in{\sf span}\{v_{m''}\mid m''\in M, w(m'')<w(m')\}$. But for all such $m''$, $v_{m''}(g) = 0$ by minimality of $m$. Therefore, $v_{m'}(g) = 0$. Therefore, $m$ is the only non-zero monomial in $f$ that receives weight $w(m)$, and so it survives $\phi$. 
\end{proof}

PIT for unambiguous circuits of low depth follows easily from Lemmas \ref{BIWA} and \ref{BIWA-map}, just by noting that the degree of the univariate polynomial $\phi(f)$ (as above) is $\poly(n, d, s)^{\Delta}$. To check if a univariate is identically zero, we can query it on ${\sf degree}+1$ points in $\mathbb{F}$ (recall that a non-zero univariate has at most degree many roots). If $\Delta$ is polylogarithmic in $n, d, s$, we get the quasipolyomial running time bound. We record this as our next theorem. 

\unambiguousHS*

Theorems \ref{commutative-hittingset} and \ref{noncommutative-hittingset} are a direct consequence of combining the reduction from nonassociative circuits to commutative, associative, unambiguous circuits and the hitting set for unambiguous circuits.

\begin{proof}[Proof of Theorem \ref{commutative-hittingset}]

    To embed the hitting set from Theorem \ref{unambiguous-hittingset} into the algebra $\mathbb{C}_d$, we let $Z=\{z_{i, j, k}\mid i\in [n], j\in[d]\}$ as in Section \ref{blackbox-commutative}. Let $\mathbb{C}_d$ and $\mathbb{C'_d}$ be the algebras defined in Section \ref{blackbox-commutative}, over the extension field $\mathbb{F}(Z)$. Recall the map $\Phi:\cna\rightarrow \mathbb{C'}_d$ as defined in Section \ref{reduction}. $\Phi$ sends $x_i$ to $Z_i\in\mathbb{C'}_d$ for each $i\in[n]$ and each $j, k\in[d]$, the $(j, j+1, k)$-th entry of $Z_i$ is one and all other entries are zero. We will embed $\Phi$ into a map $\Phi':\cna \rightarrow\mathbb{C}_d$ defined as follows: we let $\Phi'$ map $x_i$ to $(Z_i, 0)$, $Z_i\in\mathbb{C}'_d$, $c\in\mathbb{F}$ to $(\mathbf{0}, c)$ (where $\mathbf{0}\in \mathbb{C}'_d$) and extend it to all of $\mathbb{C}_d$ by linearity. Note that for a polynomial $f\in \cna$ with constant term $c$ we have that $\Phi'(f)=(\Phi(f-c), c)$.

    Let $\Psi$ be a commutative, nonassociative circuit of size $s$ and product depth $\Delta$ computing a polynomial $f\in\cna$, with ${\sf degree}(f)=d'\leq d$. Consider the image $\Phi'(f) = (\Phi(f-c), c)$ of $f$ under $\Phi'$, where $c$ is the constant term in $f$. We know from Lemma \ref{set-mult} that the $(1, d'+1, 1)$-th entry of $\Phi(f-c)$ is $\phi(f_{d'})$ (see Section \ref{reduction} for details). By Lemma \ref{unambiguous-construction-commutative}, $\phi(f_{d'})$ has an unambiguous circuit $\Psi'$ of size $\leq3d^4s$ and depth $\leq \Delta$. Note that $\Phi(f_{d'})$ is a polynomial in the $d^2n$ many $Z$ variables. Therefore, embedding the hitting set $H_{3d^4s, d^2n, d, \Delta}$ on the $Z$ variables into the map $\Phi'$ gives us the theorem.
\end{proof}

The proof of Theorem \ref{noncommutative-hittingset} is also exactly the same as above, except in this case we will obtain a hitting set with elements in $(\mathbb{A}_d)^n$.


\section{Discussion}\label{sec:discuss}

Two interesting question that stem from our work are the following:
\begin{enumerate}
    \item \emph{Depth reduction for unambiguous circuits:} As briefly mentioned in the introduction, we leave open the question of whether unambiguous circuits can be depth reduced without too much blowup in size. More precisely, suppose we have an unambiguous circuit $\Psi$ of size $s$, computing a polynomial $f\in\mathbb{F}[z_1, \ldots, z_n]$ of degree $d$. Does there exist another unambiguous circuit $\Psi'$ computing $f$, with size quasipolynomial in $n,s,d$ and depth polylogarithmic in $n,s,d$? A positive answer to this question would imply that the size of the hitting sets from Theorem \ref{commutative-hittingset} and Theorem \ref{noncommutative-hittingset} can be improved to quasipolynomial, irrespective of the depth of the circuit. As far as we know, standard depth reduction techniques due to Valiant et al.\ \cite{VSBR83} and Brent \cite{Brent74} do not preserve unambiguity of circuits, even if we allow quasipolynomial blowup in size.
    \item \emph{Tightness of the dimension of $\mathbb{C}_d$ and $\mathbb{A}_d$:} Consider the following, purely mathematical question: What is the smallest possible dimension $k(d)$ of a unital algebra that does not satisfy \emph{any} identity $f\in\cna$ of degree $\leq d$? We show $k(d)\leq d(d+1)^2+1$. One could also ask the analogous question over the algebra $\ncna$. The Amitsur-Levitzki Theorem (Theorem \ref{thm:amitsur-Levitzki}) gives such a tight characterization over $\mathbb{F}\langle X\rangle$, for matrix algebras.
\end{enumerate}

\bibliographystyle{plainurl}
\bibliography{ref}

\appendix 

\end{document}